\documentclass{article}   
\usepackage[utf8]{inputenc}
\usepackage[T1]{fontenc}
\usepackage{mathscinet}
\usepackage{graphicx}
\graphicspath{{figures/}}
\usepackage{booktabs}
\usepackage{hyperref}
\usepackage{paralist}

\usepackage[dvipsnames]{xcolor}
\usepackage{tikz}
\usetikzlibrary{fit}
\usetikzlibrary{intersections}
\usetikzlibrary{decorations.pathreplacing}
\usetikzlibrary{calc,math}

\usepackage[vlined, ruled, linesnumbered]{algorithm2e}
\DontPrintSemicolon

\usepackage{amsthm}
\usepackage{amsmath}
\usepackage{amssymb}

\usepackage{thmtools}
\usepackage{thm-restate}

\usepackage{lineno}
\usepackage{xspace}
\usepackage[color=green!20]{todonotes}
\usepackage{subcaption}

\theoremstyle{plain}
\newtheorem{theorem}{Theorem}
\newtheorem{observation}[theorem]{Observation}
\newtheorem{lemma}[theorem]{Lemma}
\newtheorem{corollary}[theorem]{Corollary}
\newtheorem*{thm*}{Theorem}

\newcommand{\naesat}{\textsc{Not-All-Equal Sat}\xspace}
\newcommand{\naeksat}[1]{\textsc{Not-All-Equal #1-Sat}\xspace}
\newcommand{\posnaesat}{\textsc{Positive Not-All-Equal Sat}\xspace}
\newcommand{\posnaeksat}[1]{\textsc{Positive Not-All-Equal #1-Sat}\xspace}
\newcommand{\planaeksat}[1]{\textsc{Strongly Planar Not-All-Equal #1-Sat}\xspace}
\newcommand{\planaesat}{\textsc{Strongly Planar Not-All-Equal Sat}\xspace}

\newcommand{\restrictedsat}{\textsc{Restricted Strongly Planar Positive Not-All-Equal 3-Sat}\xspace}

\def\cC{\mathcal{C}\xspace}
\def\cV{\mathcal{V}\xspace}
\def\N{\mathbb{N}\xspace}
\def\I{\mathcal{I}\xspace}
\def\O{\mathcal{O}\xspace}

\def\B{\textbf{B}\xspace}
\DeclareMathOperator{\gndi}{gndi}
\DeclareMathOperator{\gap}{gap}
\DeclareMathOperator{\surj}{surj}
\DeclareRobustCommand{\stirling}{\genfrac\{\}{0pt}{}}

\usepackage[a4paper, total={6in, 8in}]{geometry}

\begin{document}

\title{Intersecting edge distinguishing colorings of hypergraphs}
\author{Karolina Okrasa \and Paweł Rzążewski}
\date{Faculty of Mathematics and Information Sciences,\\Warsaw University of Technology, Warsaw, Poland\\ \ \texttt{okrasak@student.mini.pw.edu.pl, p.rzazewski@mini.pw.edu.pl}}

\maketitle
\begin{abstract}
An edge labeling of a graph distinguishes neighbors by sets (multisets, resp.), if for any two adjacent vertices $u$ and $v$ the sets (multisets, resp.) of labels appearing on edges incident to $u$ and $v$ are different. In an analogous way we define total labelings distinguishing neighbors by sets or multisets: for each vertex, we consider labels on incident edges and the label of the vertex itself.

In this paper we show that these problems, and also other problems of similar flavor, admit an elegant and natural generalization as a hypergraph coloring problem. An ieds-coloring (iedm-coloring, resp.) of a hypergraph is a vertex coloring, in which the sets (multisets, resp.) of colors, that appear on every pair of intersecting edges are different.
We show upper bounds on the size of lists, which guarantee the existence of an ieds- or iedm-coloring, respecting these lists. The proof is essentially a randomized algorithm, whose expected time complexity is polynomial.
As corollaries, we derive new results concerning the list variants of graph labeling problems, distinguishing neighbors by sets or multisets. We also show that our method is robust and can be easily extended for different, related problems.

We also investigate a close connection between edge labelings of bipartite graphs, distinguishing neighbors by sets, and the so-called property \textbf{B} of hypergraphs. We discuss computational aspects of the problem  and present some classes of bipartite graphs, which admit such a labeling using two labels.
\end{abstract}


\section{Introduction} \label{sec:intro}
Among the variants of graph coloring, there is a prominent family of problems, where the coloring of vertices of $G$ is not given explicitly, but derived from some other function. Usually, this function is some labeling of edges or vertices of $G$, and the color of the vertices is based on the labels assigned to incident edges or adjacent vertices (or both in case of total labeling).

Perhaps the most famous problem of this kind was proposed by Karoński, Łuczak, and Thomason~\cite{KARONSKI2004151}. They were labeling the edges of a graph $G$ with integers $\{1,2,\ldots,k\}$, so that for every two adjacent vertices $u,v$, the {\em sums} of labels assigned to edges incident to each of $u$ and $v$ are different.  We call such a labeling {\em neighbor sum distinguishing}. Observe a neighbor sum distinguishing labeling exists if and only if $G$ has no isolated edge, such graphs will be called {\em nice}. Karoński, Łuczak, and Thomason showed that each nice graph has a neighbor sum distinguishing labeling with 183 labels, if we allow real numbers as labels. They also showed that if the minimum degree of $G$ is large, 30 (real) labels suffice, and conjectured that every nice graph has a neighbor sum distinguishing  labeling with labels $\{1,2,3\}$. This problem, known as the {\em 1-2-3 conjecture}, raised significant interest in the graph theory community. 
Addario-Berry, Dalal, McDiarmid, Reed, and Thomason~\cite{Addario-Berry2007} showed that integer labels $\{1,2,\ldots,30\}$ are sufficient to find a neighbor sum distinguishing labeling of any nice graph. The upper bound on the largest label was subsequently improved: to 16 by Addario-Berry, Dalal, and Reed~\cite{ADDARIOBERRY20081168}, then to 13 by Wang and Yu~\cite{Wang2008}, and then to 6 by Kalkowski, Karoński, and Pfender \cite{KALKOWSKI6}. Currently best bound for general graphs is 5 and was shown by Kalkowski, Karoński, and Pfender~\cite{KALKOWSKI2010347}.
Bartnicki, Grytczuk, and Niwczyk \cite{DBLP:journals/jgt/BartnickiGN09} proposed a stronger conjecture, that for any assignment of 3-element lists to edges of $G$, one can find a neighbor sum distinguishing labeling, such that every edge gets a label from its list. A constant bound on the size of lists that guarantee the existence of such a labeling is known for some special classes of graphs, e.g. complete graphs, complete bipartite graphs, or nice trees \cite{DBLP:journals/jgt/BartnickiGN09}.

Dudek and Wajc \cite{DudekDMTCS} considered the computational problem of deciding whether a given nice graph has a neighbor sum distinguishing edge labeling with labels $\{a,b\}$. They proved that the problem is NP-complete for $\{a,b\}=\{0,1\}$ and for $\{a,b\}=\{1,2\}$. This was later extended by Dehghan, Sadeghi, and Ahadi \cite{DEHGHAN201325} to all pairs $\{a,b\}$, even if the input graph is cubic.

Let us have a closer look at the already mentioned result by Karoński, Łuczak, and Thomason~\cite{KARONSKI2004151}, that a constant number of real labels is sufficient to distinguish adjacent vertices by sums of labels on incident edges.
We will say that an edge labeling {\em distinguishes neighbors by multisets}, if for every pair $u,v$ of adjacent vertices, the 
{\em multisets} of labels appearing on the edges incident to $u$ and $v$ are different, see \autoref{fig:edge-lab} a). 
Notice that different sums always imply different multisets, but it is possible to have different multisets that give the same sum.
The authors of \cite{KARONSKI2004151} proved that every graph has an edge labeling, which distinguishes neighbors by multisets, and uses a constant number of labels and then, by choosing labels that satisfy certain independence properties, one can obtain a labeling that distinguishes sums.
Edge labelings distinguishing neighbors by multisets were further studied by Addario-Berry, Aldred, Dalal, and Reed~\cite{ADDARIOBERRY2005237}, who showed every nice graph has such a labeling using four labels and,  if the minimum degree of $G$ is large enough, three labels suffice. 
Observe that if $G$ is regular, then an edge labeling using two labels is neighbor sum distinguishing if and only if it distinguishes neighbors by multisets. Thus, by the already mentioned result by Dehghan, Sadeghi, and Ahadi \cite{DEHGHAN201325}, it is NP-complete to decide whether a given graph has an edge labeling distinguishing neighbors by multisets.

Instead of considering sums or multisets of labels appearing on edges incident to adjacent vertices, one can also consider {\em sets}. We say that an edge labeling {\em distinguishes neighbors by sets} if for any two adjacent vertices $v$ and $w$, the sets of labels on edges incident to $u$ and $v$ are different, see \autoref{fig:edge-lab} b).
Observe that different sets of colors always imply different multisets, but not the other way around.
On the other hand, sets are sums are incomparable: one may have different multisets that give the same sum and different sets, or the same set and different sums.
By the {\em generalized neighbor-distinguishing index} of a graph $G$, denoted by $\gndi(G)$, we mean the minimum number of labels in an edge labeling distinguishing neighbors by sets. This parameter was introduced by Gy\H{o}ri, Hor\v{n}\'ak, Palmer, and Woźniak \cite{GYORI2008827}, who proved that for every nice graph $G$ we have $\gndi(G) \leq 2\lceil \log_2 \chi(G) \rceil +1$. This bound was later refined by Hor\v{n}ak and Sotak \cite{DBLP:journals/dm/HornakS10} and by Gy\H{o}ri and Palmer~\cite{DBLP:journals/dm/GyoriP09}, who proved that if $\chi(G) \geq 3$, then $\gndi(G) = \lceil \log_2 \chi(G) \rceil +1$. Hor\v{n}\'ak and Woźniak considered a list variant of the problem, where each edge is equipped with a list of possible labels and we ask for the existence of a labeling distinguishing neighbors by multisets and respecting these lists. They proved tight bounds on the size of lists, that guarantee the existence of such a labeling in paths and cycles, and showed that for trees lists of size three are sufficient. They also showed that lists of size three may be necessary, even for trees $T$ with $\gndi(T)=2$.

\begin{figure}[h]
\begin{center}
\begin{subfigure}[b]{0.5\textwidth}
\begin{tikzpicture}
\node at (-1,2.5) {(a)};
\draw[very thick, draw=blue!50] (1,0)--(2,0);
\draw[very thick, draw=red!50] (2,0)--(3,0);
\draw[very thick, draw=red!50] (1,2)--(2,2);
\draw[very thick, draw=red!50] (2,2)--(3,2);
\draw[very thick, draw=red!50] (1,0)--(1,2);
\draw[very thick, draw=red!50] (2,0)--(2,2);
\draw[very thick, draw=red!50] (3,0)--(3,2);
\draw[very thick, draw=blue!50] (1,0)--(0,1);
\draw[very thick, draw=blue!50] (0,1)--(1,2);
\draw[very thick, draw=red!50] (4,1)--(3,0);
\draw[very thick, draw=blue!50] (4,1)--(3,2);
\draw[very thick, draw=blue!50] (0,1)--(4,1);
\filldraw[fill=white] (0,1) circle(4pt);
\filldraw[fill=white] (1,0) circle(4pt);
\filldraw[fill=white] (3,0) circle(4pt);
\filldraw[fill=white] (2,0) circle(4pt);
\filldraw[fill=white] (2,2) circle(4pt);
\filldraw[fill=white] (1,2) circle(4pt);
\filldraw[fill=white] (3,2) circle(4pt);
\filldraw[fill=white] (4,1) circle(4pt);
\draw [fill=blue!50] (-0.4,0.9) rectangle (-0.2,1.1);
\draw [fill=blue!50] (-0.4,1.1) rectangle (-0.2,1.3);
\draw [fill=blue!50] (-0.4,0.7) rectangle (-0.2,0.9);

\draw [fill=red!50] (0.9,-0.4) rectangle (1.1,-0.2);
\draw [fill=blue!50] (0.9,-0.6) rectangle (1.1,-0.4);
\draw [fill=blue!50] (0.9,-0.8) rectangle (1.1,-0.6);

\draw [fill=red!50] (1.9,-0.4) rectangle (2.1,-0.2);
\draw [fill=red!50] (1.9,-0.6) rectangle (2.1,-0.4);
\draw [fill=blue!50] (1.9,-0.8) rectangle (2.1,-0.6);

\draw [fill=red!50] (2.9,-0.4) rectangle (3.1,-0.2);
\draw [fill=red!50] (2.9,-0.6) rectangle (3.1,-0.4);
\draw [fill=red!50] (2.9,-0.8) rectangle (3.1,-0.6);

\draw [fill=blue!50] (0.9,2.4) rectangle (1.1,2.2);
\draw [fill=red!50] (0.9,2.6) rectangle (1.1,2.4);
\draw [fill=red!50] (0.9,2.8) rectangle (1.1,2.6);

\draw [fill=red!50] (1.9,2.4) rectangle (2.1,2.2);
\draw [fill=red!50] (1.9,2.6) rectangle (2.1,2.4);
\draw [fill=red!50] (1.9,2.8) rectangle (2.1,2.6);

\draw [fill=blue!50] (2.9,2.4) rectangle (3.1,2.2);
\draw [fill=red!50] (2.9,2.6) rectangle (3.1,2.4);
\draw [fill=red!50] (2.9,2.8) rectangle (3.1,2.6);

\draw [fill=blue!50] (4.4,0.9) rectangle (4.2,1.1);
\draw [fill=red!50] (4.4,1.1) rectangle (4.2,1.3);
\draw [fill=blue!50] (4.4,0.7) rectangle (4.2,0.9);
\end{tikzpicture}
\end{subfigure}%
\begin{subfigure}[b]{0.5\textwidth}
\begin{tikzpicture}
\node at (-1,2.5) {(b)};
\draw[very thick, draw=blue!50] (1,0)--(2,0);
\draw[very thick, draw=red!50] (2,0)--(3,0);
\draw[very thick, draw=red!50] (1,2)--(2,2);
\draw[very thick, draw=red!50] (2,2)--(3,2);
\draw[very thick, draw=green!50] (1,0)--(1,2);
\draw[very thick, draw=red!50] (2,0)--(2,2);
\draw[very thick, draw=green!50] (3,0)--(3,2);
\draw[very thick, draw=blue!50] (1,0)--(0,1);
\draw[very thick, draw=blue!50] (0,1)--(1,2);
\draw[very thick, draw=red!50] (4,1)--(3,0);
\draw[very thick, draw=blue!50] (4,1)--(3,2);
\draw[very thick, draw=blue!50] (0,1)--(4,1);

\filldraw[fill=white] (0,1) circle(4pt);
\filldraw[fill=white] (1,0) circle(4pt);
\filldraw[fill=white] (3,0) circle(4pt);
\filldraw[fill=white] (2,0) circle(4pt);
\filldraw[fill=white] (2,2) circle(4pt);
\filldraw[fill=white] (1,2) circle(4pt);
\filldraw[fill=white] (3,2) circle(4pt);
\filldraw[fill=white] (4,1) circle(4pt);

\draw [fill=blue!50] (-0.4,0.9) rectangle (-0.2,1.1);

\draw [fill=green!50] (0.9,-0.4) rectangle (1.1,-0.2);
\draw [fill=blue!50] (0.9,-0.6) rectangle (1.1,-0.4);

\draw [fill=red!50] (1.9,-0.4) rectangle (2.1,-0.2);
\draw [fill=blue!50] (1.9,-0.6) rectangle (2.1,-0.4);

\draw [fill=red!50] (2.9,-0.4) rectangle (3.1,-0.2);
\draw [fill=green!50] (2.9,-0.6) rectangle (3.1,-0.4);

\draw [fill=blue!50] (0.9,2.4) rectangle (1.1,2.2);
\draw [fill=green!50] (0.9,2.6) rectangle (1.1,2.4);
\draw [fill=red!50] (0.9,2.8) rectangle (1.1,2.6);

\draw [fill=red!50] (1.9,2.4) rectangle (2.1,2.2);

\draw [fill=blue!50] (2.9,2.4) rectangle (3.1,2.2);
\draw [fill=green!50] (2.9,2.6) rectangle (3.1,2.4);
\draw [fill=red!50] (2.9,2.8) rectangle (3.1,2.6);

\draw [fill=blue!50] (4.4,0.9) rectangle (4.2,1.1);
\draw [fill=red!50] (4.4,1.1) rectangle (4.2,1.3);
\end{tikzpicture}
\end{subfigure}
\caption{Edge labeling distinguishing neighbors (a) by multisets, (b) by sets. The boxes next to each vertex denote (a) the multiset and (b) the set of labels appearing on edges incident to that vertex.}
\label{fig:edge-lab}
\end{center}
\end{figure}
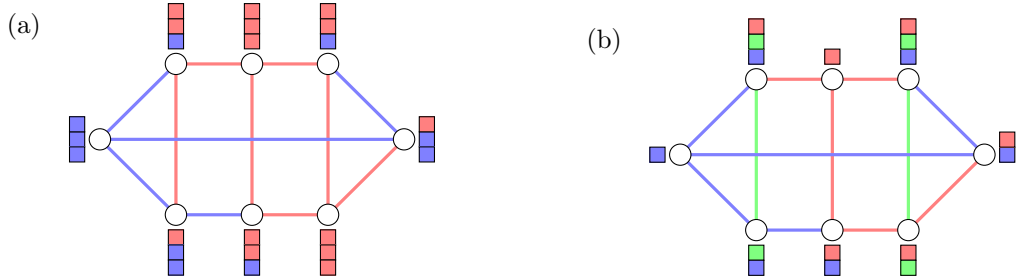

Inspired by the 1-2-3 conjecture, Woźniak and Przybyło~\cite{DBLP:journals/dmtcs/PrzybyloW10} suggested a closely related problem considering total labelings of $G$, i.e., labelings of edges and vertices. They considered the problem of finding a total labeling using minimum number of labels, which distinguishes adjacent vertices by {\em sums} of labels appearing on incident edges and the vertex itself. Such a labeling is called {\em neighbor sum distinguishing total labeling}.
Woźniak and Przybyło conjectured that every graph has a neighbor sum distinguishing total labeling, using labels $\{1,2\}$ only, this problem is known as the {\em 1-2 conjecture}.
Kalkowski~\cite{Kalkowski12} showed that each graph has a neighbor sum distinguishing total labeling with labels $\{1,2,3\}$, in which the label 3 does not appear on any vertex. 
Wong and Zhu~\cite{DBLP:journals/jgt/WongZ11}, and Przybyło and Woźniak~\cite{DBLP:journals/combinatorics/PrzybyloW11}  conjectured that the 1-2 conjecture holds even in the list variant. Recently, Wong and Zhu~\cite{Wong2016} showed a list version of the theorem by Kalkowski: a list neighbor sum distinguishing total labeling exists if each vertex  has a list of size 2 and every edge has a list of size 3.
Observe that in an analogous way one may define total labelings {\em distinguishing neighbors by multisets}  and {\em distinguishing neighbors by sets}, see \autoref{fig:total}.

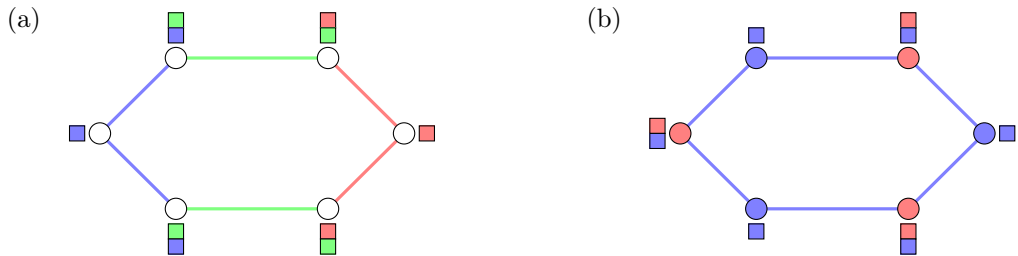
\begin{figure}[h]
\begin{center}
\begin{subfigure}[b]{0.5\textwidth}
\begin{tikzpicture}
\node at (-1,2.5) {(a)};
\draw[very thick, draw=blue!50] (1,0)--(0,1);
\draw[very thick, draw=blue!50] (0,1)--(1,2);
\draw[very thick, draw=green!50] (1,2)--(3,2);
\draw[very thick, draw=red!50] (4,1)--(3,2);
\draw[very thick, draw=red!50] (4,1)--(3,0);
\draw[very thick, draw=green!50] (1,0)--(3,0);
\filldraw[fill=white] (1,0) circle(4pt);
\filldraw[fill=white] (0,1) circle(4pt);
\filldraw[fill=white] (3,0) circle(4pt);
\filldraw[fill=white] (1,2) circle(4pt);
\filldraw[fill=white] (3,2) circle(4pt);
\filldraw[fill=white] (4,1) circle(4pt);

\draw [fill=blue!50] (-0.4,0.9) rectangle (-0.2,1.1);

\draw [fill=red!50] (4.2,0.9) rectangle (4.4,1.1);

\draw [fill=green!50] (0.9,-0.4) rectangle (1.1,-0.2);
\draw [fill=blue!50] (0.9,-0.6) rectangle (1.1,-0.4);

\draw [fill=green!50] (0.9,2.4) rectangle (1.1,2.6);
\draw [fill=blue!50] (0.9,2.2) rectangle (1.1,2.4);

\draw [fill=red!50] (2.9,-0.4) rectangle (3.1,-0.2);
\draw [fill=green!50] (2.9,-0.6) rectangle (3.1,-0.4);

\draw [fill=red!50] (2.9,2.4) rectangle (3.1,2.6);
\draw [fill=green!50] (2.9,2.2) rectangle (3.1,2.4);
\end{tikzpicture}
\end{subfigure}%
\begin{subfigure}[b]{0.5\textwidth}
\begin{tikzpicture}
\node at (-1,2.5) {(b)};
\draw[very thick, draw=blue!50] (1,0)--(0,1);
\draw[very thick, draw=blue!50] (0,1)--(1,2);
\draw[very thick, draw=blue!50] (1,2)--(3,2);
\draw[very thick, draw=blue!50] (4,1)--(3,2);
\draw[very thick, draw=blue!50] (4,1)--(3,0);
\draw[very thick, draw=blue!50] (1,0)--(3,0);

\filldraw[fill=blue!50] (1,0) circle(4pt);
\filldraw[fill=red!50] (0,1) circle(4pt);
\filldraw[fill=red!50] (3,0) circle(4pt);
\filldraw[fill=blue!50] (1,2) circle(4pt);
\filldraw[fill=red!50] (3,2) circle(4pt);
\filldraw[fill=blue!50] (4,1) circle(4pt);

\draw [fill=red!50] (-0.4,1) rectangle (-0.2,1.2);
\draw [fill=blue!50] (-0.4,1) rectangle (-0.2,0.8);

\draw [fill=blue!50] (4.2,0.9) rectangle (4.4,1.1);

\draw [fill=blue!50] (0.9,-0.4) rectangle (1.1,-0.2);

\draw [fill=blue!50] (0.9,2.2) rectangle (1.1,2.4);

\draw [fill=red!50] (2.9,-0.4) rectangle (3.1,-0.2);
\draw [fill=blue!50] (2.9,-0.6) rectangle (3.1,-0.4);

\draw [fill=red!50] (2.9,2.4) rectangle (3.1,2.6);
\draw [fill=blue!50] (2.9,2.2) rectangle (3.1,2.4);
\end{tikzpicture}
\end{subfigure}
\caption{(a) Edge and (b) total labeling distinguishing neighbors by sets.}
\label{fig:total}
\end{center}
\end{figure}

Finally, let us mention a similar problem, considered by Seamone and Stevens~\cite{DBLP:journals/dmtcs/SeamoneS13}.
Let $G$ be a graph and suppose that its edges are linearly ordered. A labeling of edges {\em distinguishes neighbors by sequences}, if the {\em sequences} of colors (implied by the global ordering of edges), appearing on edges incident to adjacent vertices, are different.

Seamone and Stevens showed that if the ordering of edges can be chosen, then for any nice graph lists of size 2 are sufficient to find a list edge labeling, distinguishing neighbors by sequences. If the ordering of edges is fixed, the lists of size 3 suffice, provided that the minimum degree is large enough, compared to the maximum degree. In particular, lists of size 3 are sufficient for a $k$-regular graph with $k \geq 6$.

Another variants of the mentioned problems have also been studied. 
For example, one can ask for an edge-labeling, which distinguishes neighbors by sums/multisetes/sums, and is also required to be proper: for distinguishing neighbors by sums of labels, see e.g. Przybyło~\cite{DBLP:journals/rsa/Przybylo15,DBLP:journals/dam/Przybylo16}, Bonamy and Przybyło~\cite{DBLP:journals/jgt/BonamyP17}, Hocquard and Przybyło~\cite{DBLP:journals/gc/HocquardP17};
for distinguishing neighbors by sets of labels, see  Zhang, Liu, and Wang~\cite{zhang2015color}, Balister, Gy\H{o}ri, Lehel, and Schelp~\cite{DBLP:journals/siamdm/BalisterGLS07},  Edwards, Hor\v{n}\'{a}k, and Woźniak~\cite{Edwards2006}, Bonamy, Bousquet, and Hocquard~\cite{10.1007/978-88-7642-475-5_50} or Hatami~\cite{hatami2005delta};
for list edge labelings distinguishing neighbors by multisets see a recent exciting result by Kwaśny and Przybyło~\cite{KwasnyPrzybylo}.

There is also some work on edge labelings, that distinguish neighbors by {\em products} of labels (see Skowronek-Kaziów~\cite{DBLP:journals/ipl/Skowronek-Kaziow08,DBLP:journals/ipl/Skowronek-Kaziow12,DBLP:journals/jco/Skowronek-Kaziow17}).
For a more detailed overview on the related problems, we refer the reader to the recent book by Zhang~\cite{zhang2015color}, and a survey by Seamone~\cite{DBLP:journals/corr/abs-1211-5122}.

\subsection{Our contribution}

Let us introduce the main character of this paper.
For a hypergraph $H$, we say that a vertex coloring is \emph{intersecting edge distinguishing by multisets} (or, in short, is an \emph{iedm-coloring}), if the \emph{multisets} of colors appearing in intersecting edges are different. 
Similarly, a vertex coloring is \emph{intersecting edge distinguishing by sets} (or, in short, is an \emph{ieds-coloring}), if the \emph{sets} of colors appearing in intersecting edges are different. 

It is perhaps interesting to mention the special case of graphs, i.e., 2-uniform hypergraphs. It is straightforward to verify that in this case iedm- and ieds-colorings are equivalent. Moreover, a vertex coloring distinguishes intersecting edges by sets (or, equivalently, multisets) if and only if no two vertices with a common neighbor receive the same color. Such a concept is already known in graph theory and usually referred to as an \emph{$L(0,1)$-labeling}. The motivation to study this kind of a coloring came from the hidden terminal problem in telecommunication~\cite{Bertossi:1995:CAH:225996.226008,DBLP:journals/tcom/Makansi87}. Optimal (i.e., using the minimum number of colors) $L(0,1)$-labelings are known for simple classes of graphs, like paths, cycles, grids (see Makansi~\cite{DBLP:journals/tcom/Makansi87} and Jin, Yeh~\cite{NAV:NAV20041}), hypercubes (see Wan~\cite{Wan1997}), and complete binary trees (see Bertossi and Bonuccelli~\cite{Bertossi:1995:CAH:225996.226008}). Bodlaender, Kloks, Tan, van Leeuwen~\cite{DBLP:journals/cj/BodlaenderKTL04} showed some bounds on the number of colors required to find an $L(0,1)$-labeling of $G$ for some special classes of graphs: bounded-treewidth graphs, permutation graphs, outerplanar graphs, split graphs, and bipartite graphs. On the complexity side, it is known that the decision problem whether an input graph has an $L(0,1)$-labeling with 3 colors is NP-complete for planar graphs~\cite{Bertossi:1995:CAH:225996.226008} and split graphs ~\cite{DBLP:journals/cj/BodlaenderKTL04}. The parameterized complexity of this problem was considered by Fiala, Golovach, and Kratochv\'il~\cite{FIALA20112513}, who showed that the problem is $W[1]$-hard, when parameterized by treewidth, but FPT, when parameterized by the vertex cover number (we refer the reader to the book by Cygan {\em et al.}~\cite{DBLP:books/sp/CyganFKLMPPS15} for more information about parameterized complexity).

In \autoref{sec:hyper} we argue that colorings of hypergraphs, that distinguishing intersecting edges by sets or multisets, are natural common generalization of edge and total labelings distinguishing neighbors by sets and multisets.
We are interested in the list variants of both problems. As the main result, in \autoref{thm:main} we show upper bounds on the size of lists that guarantee the existence of a list ieds-coloring or a list iedm-coloring of a given hypergraph.

The main part of the paper, i.e., \autoref{sec:entropy}, is devoted to the proof of \autoref{thm:main}. The proof uses the so-called entropy compression method, which is a variant of the Lov\'{a}sz Local Lemma~\cite{erdos1975problems}. The Local Lemma is essentially non-constructive, but several algorithmic versions have also been developed (see Alon~\cite{alon1991parallel}, Molloy, Reed~\cite{molloy1998further}, and Moser, Tardos~\cite{moser2010constructive}). Entropy compression originates in the algorithmic version of the Local Lemma by Moser and Tardos, and was first used by Grytczuk, Kozik, and Micek~\cite{DBLP:journals/rsa/GrytczukKM13} to study the list version of the problem of Thue. Then the method  was successfully applied in many other contexts (see e.g. Esperet, Parreau~\cite{esperet2013acyclic} or Dujmovi{\'c}, Joret, Kozik, and Wood~\cite{dujmovic2015nonrepetitive}). The main idea of our proof is similar to the one of Bosek, Czerwiński, Grytczuk, and Rzążewski~\cite{BosekAADM}, but there are two significant differences. First, the authors of~\cite{BosekAADM} considered the so-called {\em harmonious colorings}, where {\em all} pairs of edges need to be distinguished.  Second, they were considering colorings in which no edge contained two vertices in the same color. If we drop this restriction, distinguishing edges by sets is significantly more difficult and makes the argument more complicated.

The proof of \autoref{thm:main} is essentially a randomized algorithm which finds a list ieds- or iedm-coloring of a given hypergraph. In \autoref{sec:numiters} we consider its computational complexity and show that the expected number of steps in the execution of the algorithm is polynomial.

In \autoref{sec:special} we discuss the applicability of \autoref{thm:main} and the method used in the proof.
As corollaries from Theorem \ref{thm:main}, in \autoref{sec:corollaries} we obtain several bounds on the size of lists, which guarantee the existence of a list edge/total labeling, distinguishing neighbors of a regular graph by sets or multisets.
In particular, in Corollaries \ref{cor:edge-set} and \ref{cor:total-set} we show that if $k$ is sufficiently large, then lists of size $0.54k$ are sufficient to find an edge or total labeling distinguishing neighbors by sets.
In case of distinguishing by multisets, in Corollaries \ref{cor:edge-multiset} and \ref{cor:total-multiset} we prove that lists of size $0.37k$ suffice, if $k$ is large enough. We also consider the case of the so-called configurations.
In \autoref{sec:sequences}, we show that our method can be easily extended and adapted to other problem of similar kind. In particular, we show how to construct an algorithm, which finds a list coloring of vertices of a hypergraph, distinguishing intersecting edges by sequences. We omit most of the details of the proof, as it is essentially the same as the proof of \autoref{thm:main}. We focus on highlighting the crucial issues that have to be considered, when adapting our approach to a new  problem. As a side result, we improve the result of Seamone and Stevens~\cite{DBLP:journals/dmtcs/SeamoneS13} for regular graphs. In particular, we show that given a $k$-regular graph with $k \geq 10$, with a fixed ordering of edges, then lists of size 2 are sufficient to choose a list edge labeling, in which every two adjacent vertices have distinct sequences of colors appearing on incident edges.

In \autoref{sec:gndi}, we investigate an interesting relation between the generalized neighbor-distinguishing index of bipartite graphs and two well-known problems, i.e., a variant of the satisfiability problem called \naesat, and the so-called property \B of hypergraphs.
The paper is concluded in \autoref{sec:conclusion} with several open problems and suggestions for future work.

\section{Preliminaries} \label{sec:prelim}
For an integer $n$, by $[n]$ we denote the set $\{1,2,\ldots,n\}$.
By $f_n$ we denote the number of total preorders of $[n]$. In other words, $f_n$ is the number of orderings of $[n]$ with possible ties.
Observe that 
\begin{equation}\label{fubini}
f_n = \sum_{i=0}^n i! \;\stirling{n}{i},
\end{equation} where $\stirling{n}{i}$ denotes the Stirling number of the second kind, i.e., the number of partitions of $[n]$ into $i$ non-empty subsets. The value $f_n$ is sometimes called an \emph{ordered Bell number} or a \emph{Fubini number} (see the corresponding OEIS entry \cite{oeis}).
For any function $\varphi \colon X \to Y$, and any subset $X' \subseteq X$, by $\varphi\langle X' \rangle$ we denote the \emph{multiset} of images of elements of $X'$. This is in contrast with the usual notation $\varphi(X')$, which is the image of $X'$, i.e., the \emph{set} of images of elements of $X'$.
For two disjoint sets $A,B$ and a function $f \colon A \cup B \to \N$, such that $f(A) \subseteq f(B)$, we say that a function $\gamma \colon A \to B$ is \emph{color-preserving} if for every $a \in A$ it holds that $f(a) = f(\gamma(a))$.

Consider a hypergraph $H=(V,E)$. For every vertex $v \in V$, let $E(v)$ be the set of edges containing $v$. The \emph{degree} of a vertex $v$ is defined as $\deg v := |E(v)|$. By $\Delta(H)$ we denote the maximum degree, i.e., $\max_{v\in V} \deg v$. A hypergraph is \emph{$k$-uniform}, for $k \in \mathbb{N}$, if $|P|=k$ for every $P \in E$. Clearly, graphs are 2-uniform hypergraphs.
For a $k$-uniform hypergraph $H$, define 
\[
I(H):=\{|P \setminus Q|: P,Q \in E\} \cap [k-1].
\]

\subsection{Graph labeling problems} \label{sec:graphlabeling}

Let $G=(V, E)$ be a simple, undirected graph, i.e., a 2-uniform hypergraph and consider an (unrestricted) edge labeling $\mu$ of $G$. We say that $\mu$ \emph{distinguishes neighbors by sets} (respectively, \emph{by multisets}) if $\mu(E(v)) \neq \mu (E(u))$ (respectively, $\mu\langle E(v) \rangle \neq \mu \langle E(u) \rangle$) for every pair of adjacent vertices $v$ and $u$.
Clearly, in both cases, such a labeling exists if and only if and only if the graph $G$ is nice, i.e., it does not have an isolated edge. 
Recall that by $\gndi(G)$ we denote the minimum number of labels used in an edge labeling of $G$, distinguishing neighbors by sets.

Similarly, let $\eta$ be a total labeling of $G$, i.e., a labeling of its edges and vertices. We say it \emph{distinguishes neighbors by sets} (respectively, \emph{by multisets}) if $\eta(E(v)\cup \{v\}) \neq \eta (E(u)\cup\{u\})$ (respectively, $\eta\langle E(v)\cup \{v\} \rangle \neq \eta \langle E(u) \cup \{u\} \rangle$) for every pair of adjacent vertices $v$ and $u$.
Note that such a labeling exists for every graph $G$, as it is enough to use a different label for every vertex, and one extra label for all edges.

In list variants of all four problems, edges (or edges and vertices in total labelings) are equipped with lists of possible labels, and we ask for a labeling, where the label of every edge (or edge and vertex in total labelings) belongs to the appropriate list.

It is interesting to note a very close connection between total and edge labelings distinguishing neighbors by sets. Consider a nice graph $G$ and let $m$ denote the minimum number of labels in a total labeling distinguishing neighbors by sets. First, observe that $m \leq \gndi(G)$, because we can extend any edge labeling, which is neighbor distinguishing by sets, to a total labeling, by assigning to each vertex $v$ a label that appears on some edge incident to $v$.
On the other hand, consider a total labeling $c$ of $G$, which distinguishes neighbors by sets and uses $m$ labels. 
By the argument analogous to the one used by Hor\v{n}{\'{a}}k and Sot{\'{a}}k~\cite{DBLP:journals/dm/HornakS10}, we observe that
\[\{c^{-1}(M) \cup c^{-1}([m] \setminus M) \; \colon \; M \subseteq [m]\},\]
is a proper vertex coloring of $G$, which implies that $\chi(G) \leq 2^{m-1}$ and thus $m\geq\lceil\log_2\chi(G)\rceil +1$. Recall that Gy\H{o}ri and Palmer~\cite{DBLP:journals/dm/GyoriP09} showed that if $\chi(G) \geq 3$, then $\gndi(G)=\lceil\log_2\chi(G)\rceil+1$, which implies that $m = \gndi(G)$.

Finally, if $G$ is a bipartite graph then $m=2$, because we can extend its proper vertex coloring using colors $\{1,2\}$ to a total labeling, by assigning label 1 to every edge of $G$.
The inverse of this statement is also true: if $G=(V,E)$ has a total labeling $c \colon V \cup E \to \{1,2\}$, distinguishing neighbors by sets, then $G$ is bipartite. This is because the sets  $c^{-1}(\{ 1,2 \})$ and $c^{-1}(\{1 \}) \cup c^{-1}(\{2\})$ form a bipartition of $G$.

Observe that this reasoning does not show equivalence of the list variants of the problems, so it still makes sense to consider them separately.

Finally, note that if we are interested in edge/total labelings of a disconnected graph $G$, then we can label each connected component of $G$ independently, as the distinguishing constraints are local. Thus we will focus on connected graphs.

\subsection{Distinguishing neighbors via colorings of hypergraphs} \label{sec:hyper}

Let $H=(V,E)$ be a hypergraph, we do not allow multiple edges.
We say that a coloring $\varphi \colon V \rightarrow \mathbb{N}$ \emph{distinguishes intersecting edges by sets} (or, in short, is an \emph{ieds-coloring}) if for every pair of distinct intersecting edges $P, Q \in E$ it holds that
\begin{equation}
\label{def:leds-col}
\varphi(P) \neq \varphi(Q).
\end{equation}
Analogously, $\varphi$ \emph{distinguishes intersecting edges by multisets} (or is an \emph{iedm-coloring}), if for every pair of distinct intersecting edges $P, Q \in E$ it holds that
\begin{equation}
\label{def:ledm-col}
\varphi\langle P \rangle \neq \varphi\langle Q \rangle.
\end{equation}
See \autoref{fig:iedm-ieds} for an example. Note that \eqref{def:leds-col} implies \eqref{def:ledm-col}, so every ieds-coloring is also an iedm-coloring. Moreover, if $|P|\neq|Q|$, then \eqref{def:ledm-col} is satisfied for any coloring $\varphi$.  

By a \emph{list ieds-coloring} (respectively, a \emph{list iedm-coloring}) we mean an ieds-coloring (respectively, an iedm-coloring), in which the color of each vertex $v$ is chosen from a list $L_v$ which is assigned to the vertex $v$. The lists come with a graph and are assumed to be a part of the instance.

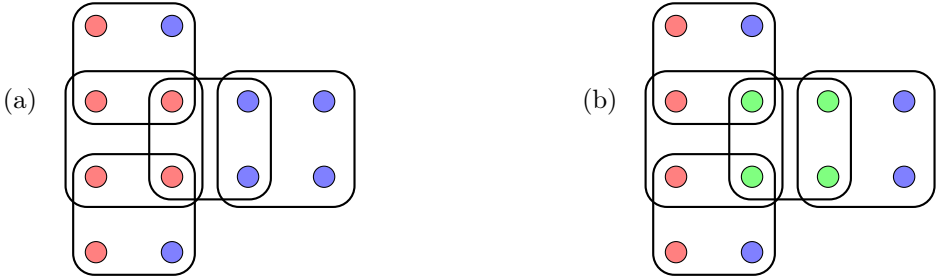
\begin{figure}[h]
\begin{center}
\begin{subfigure}[b]{0.5\textwidth}
\begin{tikzpicture}
\node at (0,1) {(a)};
\draw[thick,rounded corners=8pt] (1.5,-0.4)--(2.4,-0.4)--(2.4,1.4)--(0.6,1.4)--(0.6,-0.4)--(1.5,-0.4);
\draw[thick,rounded corners=8pt] (1.5,-1.3)--(2.3,-1.3)--(2.3,0.3)--(0.7,0.3)--(0.7,-1.3)--(1.5,-1.3);
\draw[thick,rounded corners=8pt] (1.5,0.7)--(2.3,0.7)--(2.3,2.3)--(0.7,2.3)--(0.7,0.7)--(1.5,0.7);
\draw[thick,rounded corners=8pt] (3.5,-0.4)--(4.4,-0.4)--(4.4,1.4)--(2.6,1.4)--(2.6,-0.4)--(3.5,-0.4);
\draw[thick,rounded corners=8pt] (2.5,-0.3)--(3.3,-0.3)--(3.3,1.3)--(1.7,1.3)--(1.7,-0.3)--(2.5,-0.3);
\fill[fill=white, fill opacity=0.8] (1.35,0) circle(6pt);
\filldraw[fill=red!50] (1,0) circle(4pt);
\filldraw[fill=red!50] (1,1) circle(4pt);
\filldraw[fill=red!50] (2,0) circle(4pt);
\filldraw[fill=red!50] (2,1) circle(4pt);
\filldraw[fill=blue!50] (3,0) circle(4pt);
\filldraw[fill=blue!50] (3,1) circle(4pt);
\filldraw[fill=blue!50] (4,0) circle(4pt);
\filldraw[fill=blue!50] (4,1) circle(4pt);
\filldraw[fill=red!50] (1,-1) circle(4pt);
\filldraw[fill=red!50] (1,2) circle(4pt);
\filldraw[fill=blue!50] (2,-1) circle(4pt);
\filldraw[fill=blue!50] (2,2) circle(4pt);
\end{tikzpicture}
\end{subfigure}%
\begin{subfigure}[b]{0.5\textwidth}
\begin{tikzpicture}
\node at (0,1) {(b)};
\draw[thick,rounded corners=8pt] (1.5,-0.4)--(2.4,-0.4)--(2.4,1.4)--(0.6,1.4)--(0.6,-0.4)--(1.5,-0.4);
\draw[thick,rounded corners=8pt] (1.5,-1.3)--(2.3,-1.3)--(2.3,0.3)--(0.7,0.3)--(0.7,-1.3)--(1.5,-1.3);
\draw[thick,rounded corners=8pt] (1.5,0.7)--(2.3,0.7)--(2.3,2.3)--(0.7,2.3)--(0.7,0.7)--(1.5,0.7);
\draw[thick,rounded corners=8pt] (3.5,-0.4)--(4.4,-0.4)--(4.4,1.4)--(2.6,1.4)--(2.6,-0.4)--(3.5,-0.4);
\draw[thick,rounded corners=8pt] (2.5,-0.3)--(3.3,-0.3)--(3.3,1.3)--(1.7,1.3)--(1.7,-0.3)--(2.5,-0.3);
\fill[fill=white, fill opacity=0.8] (1.35,0) circle(6pt);
\filldraw[fill=red!50] (1,0) circle(4pt);
\filldraw[fill=red!50] (1,1) circle(4pt);
\filldraw[fill=green!50] (2,0) circle(4pt);
\filldraw[fill=green!50] (2,1) circle(4pt);
\filldraw[fill=green!50] (3,0) circle(4pt);
\filldraw[fill=green!50] (3,1) circle(4pt);
\filldraw[fill=blue!50] (4,0) circle(4pt);
\filldraw[fill=blue!50] (4,1) circle(4pt);
\filldraw[fill=red!50] (1,-1) circle(4pt);
\filldraw[fill=red!50] (1,2) circle(4pt);
\filldraw[fill=blue!50] (2,-1) circle(4pt);
\filldraw[fill=blue!50] (2,2) circle(4pt);
\end{tikzpicture}
\end{subfigure}
\caption{(a) An iedm-coloring and (b) an ieds-coloring of a 4-uniform hypergraph.}
\label{fig:iedm-ieds}
\end{center}
\end{figure}

Now, for the graph $G=(V,E)$, let $H=(E,Q)$ be its \emph{dual hypergraph}, i.e., the hypergraph whose vertex set is the set of edges of $G$, and edges of $H$ correspond to vertices of $G$ in the following way: $Q=\{E(v) \; \colon \;v \in V\}$. Observe that in such a hypergraph each vertex belongs to exactly two edges and also, if $G$ is $k$-regular, then $H$ is $k$-uniform and $I(H)=\{k-1\}$. 
Let $\mu$ be an edge labeling of a graph $G$. We observe that $\mu$ distinguishes neighbors by sets (respectively, by multisets) if and only if it is an ieds-coloring (respectively, iedm-coloring) of the dual hypergraph $H$ of $G$, see \autoref{fig:dual-hyper}.

Analogously to the previous case, we define the \emph{total hypergraph} of $G$, i.e., the hypergraph $H=(V \cup E,Q)$, whose vertices are both vertices and edges of $G$. Every edge of $H$ is the set of all edges incident to a vertex in $G$ and the vertex itself, i.e., $Q=\{E(v)\cup \{v\} \; \colon \; v \in V\}$. Note that $\Delta(H)\leq 2$ and, if $G$ is $k$-regular, then $H$ is $(k+1)$-uniform and $I(H)=\{k\}$. 
Now, if $\eta$ is a total labeling of $G$, it distinguishes neighbors by sets (respectively, by multisets) if and only if it is an ieds-coloring (respectively, an iedm-coloring) of $H$.

\begin{figure}[h]
\begin{center}
\begin{subfigure}[b]{0.5\textwidth}
\centering\begin{tikzpicture}
\node at (-0.5,3) {(a)};
\draw[very thick, draw=blue!50] (0,0)--(0,2);
\draw[very thick, draw=blue!50] (0,0)--(4,0);
\draw[very thick, draw=green!50] (4,0)--(4,2);
\draw[very thick, draw=green!50] (0,2)--(4,2);
\draw[very thick, draw=green!50] (0,0)--(1,1);
\draw[very thick, draw=red!50] (0,2)--(1,1);
\draw[very thick, draw=red!50] (4,0)--(3,1);
\draw[very thick, draw=blue!50] (4,2)--(3,1);
\draw[very thick, draw=red!50] (1,1)--(3,1);
\filldraw[fill=white] (1,1) circle(4pt) node[anchor=east] {$a \ $};
\filldraw[fill=white] (3,1) circle(4pt) node[anchor=west] {$\ d$};
\filldraw[fill=white] (0,2) circle(4pt) node[anchor=east] {$b \ $};
\filldraw[fill=white] (0,0) circle(4pt) node[anchor=east] {$c \ $};
\filldraw[fill=white] (4,0) circle(4pt) node[anchor=west] {$\ e$};
\filldraw[fill=white] (4,2) circle(4pt) node[anchor=west] {$\ f$};

\draw [fill=blue!50] (-0.1,2.2) rectangle (0.1,2.4);
\draw [fill=green!50] (-0.1,2.4) rectangle (0.1,2.6);
\draw [fill=red!50] (-0.1,2.8) rectangle (0.1,2.6);

\draw [fill=green!50] (-0.1,-0.2) rectangle (0.1,-0.4);
\draw [fill=blue!50] (-0.1,-0.4) rectangle (0.1,-0.6);

\draw [fill=blue!50] (3.9,2.2) rectangle (4.1,2.4);
\draw [fill=green!50] (3.9,2.4) rectangle (4.1,2.6);

\draw [fill=red!50] (3.9,-0.2) rectangle (4.1,-0.4);
\draw [fill=green!50] (3.9,-0.4) rectangle (4.1,-0.6);
\draw [fill=blue!50] (3.9,-0.6) rectangle (4.1,-0.8);

\draw [fill=green!50] (0.9,1.2) rectangle (1.1,1.4);
\draw [fill=red!50] (0.9,1.4) rectangle (1.1,1.6);

\draw [fill=blue!50] (2.9,1.2) rectangle (3.1,1.4);
\draw [fill=red!50] (2.9,1.4) rectangle (3.1,1.6);
\end{tikzpicture}
\end{subfigure}
\begin{subfigure}[b]{0.45\textwidth}
\centering\begin{tikzpicture}
\node at (-0.5,2) {(b)};
\draw[thick,rounded corners=10pt] (0.4,0)--(0.4,-1.5)--(1.4,-1.5)--(1.4,3.5)--(0.4,3.5)--(0.4,0);
\draw[thick,rounded corners=10pt] (4.4,0)--(4.4,-1.5)--(5.4,-1.5)--(5.4,3.5)--(4.4,3.5)--(4.4,0);
\draw[thick,rounded corners=10pt] (2,-1.5)--(0.3,-1.5)--(0.3,-0.5)--(5.5,-0.5)--(5.5,-1.5)--(2,-1.5);
\draw[thick,rounded corners=10pt] (2,2.5)--(0.3,2.5)--(0.3,3.5)--(5.5,3.5)--(5.5,2.5)--(2,2.5);
\draw[thick,rounded corners=15pt] (3.4,2)--(3.4,3.7)--(0.3,0.7)--(3.4,0.7)--(3.4,2);
\draw[thick,rounded corners=15pt] (2.6,0)--(2.6,-1.7)--(5.7,1.3)--(2.6,1.3)--(2.6,0);
\fill[fill=white, fill opacity=0.8] (2.6,3) circle(7pt);
\fill[fill=white, fill opacity=0.8] (2.6,1) circle(7pt);
\fill[fill=white, fill opacity=0.8] (2.6,-1) circle(7pt);
\fill[fill=white, fill opacity=0.8] (0.6,3) circle(7pt);
\fill[fill=white, fill opacity=0.8] (0.6,1) circle(7pt);
\fill[fill=white, fill opacity=0.8] (0.6,-1) circle(7pt);
\fill[fill=white, fill opacity=0.8] (4.4,3) circle(7pt);
\fill[fill=white, fill opacity=0.8] (4.4,1) circle(7pt);
\fill[fill=white, fill opacity=0.8] (4.4,-1) circle(7pt);
\filldraw[fill=blue!50] (1,-1) circle(4pt) node[anchor=east] {$ce \ $};
\filldraw[fill=green!50] (1,1) circle(4pt) node[anchor=east] {$ac \ $};
\filldraw[fill=blue!50] (1,3) circle(4pt) node[anchor=east] {$bc \ $};
\filldraw[fill=red!50] (3,-1) circle(4pt) node[anchor=east] {$de \ $};
\filldraw[fill=red!50] (3,1) circle(4pt) node[anchor=east] {$ad \ $};
\filldraw[fill=red!50] (3,3) circle(4pt) node[anchor=east] {$ab \ $};
\filldraw[fill=green!50] (5,-1) circle(4pt) node[anchor=east] {$ef \ $};
\filldraw[fill=blue!50] (5,1) circle(4pt) node[anchor=east] {$df \ $};
\filldraw[fill=green!50] (5,3) circle(4pt) node[anchor=east] {$bf \ $};
\end{tikzpicture}
\end{subfigure}
\caption{An edge labeling of a graph, distinguishing neighbors by sets (a), and the corresponding ieds-coloring of its dual hypergraph (b).}
\label{fig:dual-hyper}
\end{center}
\end{figure}
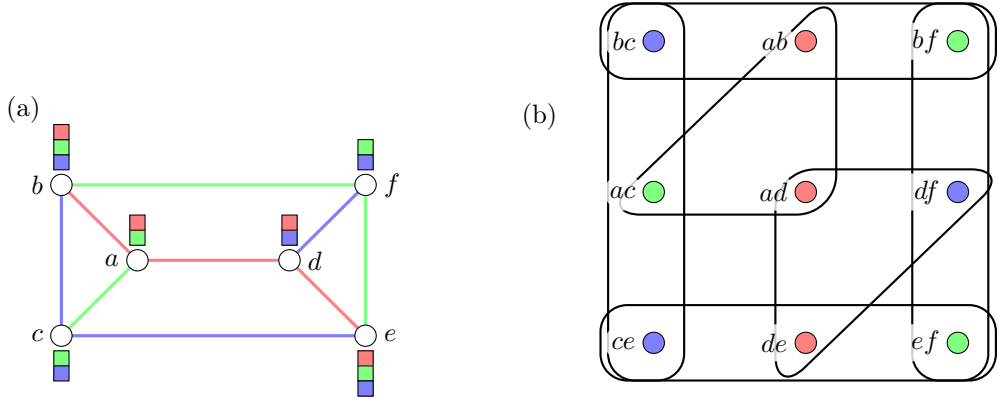

\section{Upper bound} \label{sec:entropy}
In this section we prove the main results of the paper, i.e., upper bounds on the size of lists, which guarantee the existence of an ieds-coloring or an iedm-coloring of a given $k$-uniform hypergraph.

\begin{theorem}
\label{thm:main}
Let $k \geq 2$ and let $H$ be a $k$-uniform hypergraph $H$ with $\Delta(H) \leq \Delta$ and $I(H)=I$, whose every vertex is equipped with a list $L_v$ of at least $R$ colors.
Define $q_1:=1$ and $q_i:= \frac{i}{i-1}\sqrt[i]{i-1}$ for every $i > 1$.
The following hold:
\begin{enumerate}[(a)]
\item $\begin{aligned}[t] \text{ if } 
R \geq \left\lceil 2 + \sum_{i \in I} q_i \sqrt[i]{\frac{\Delta(\Delta-1)(k-1)}{k-i} 2^{k-i+1} f_i}\right\rceil, &
 \text{ then there exists a list ieds-coloring $\varphi$ of $H$,} \end{aligned}$ \label{thmain-ieds}
\item $\begin{aligned}[t] \text{ if } 
R \geq \left\lceil 2 + \sum_{i \in I} q_i \sqrt[i]{\frac{\Delta(\Delta-1)(k-1)}{k-i} i!}\right\rceil, &
 \text{ then there exists a list iedm-coloring $\varphi$ of $H$.} \end{aligned}$ \label{thmain-iedm}
\end{enumerate}
\end{theorem}

For the rest of the section, let $H=(V,E)$ be a fixed $k$-uniform hypergraph  on $n$ vertices, such that $I(H)=I$ and $\Delta(H) \leq \Delta$ for some $\Delta \geq 2$. Note that if $\Delta(H) \leq 1$, then there are no intersecting edges and thus the problem is trivial.
We assume that sets $V$ and $E$ are linearly ordered, clearly these orderings induce a linear ordering on any subset $Y$ of $V$ or $E$.
For every $y \in Y$, denote by $n_Y(y)$ a position of element $y$ in the set $Y$, determined by this ordering.
Suppose that every vertex $v$ of $H$ is assigned with a list $L_v$ of at least $R$ colors.

For a subset $U$ of $V$, a function $\varphi \colon U \to \mathbb{N}$ is called a \emph{partial ieds-coloring} if for every pair of distinct, intersecting edges $P, Q \in E$ such that $(P \cup Q) \setminus (P \cap Q) \subseteq U$, it holds that
\begin{equation}
\label{def:partial leds-col}
\varphi( P \cap U ) \neq \varphi( Q \cap U )
\end{equation}
and $\varphi(v) \in L_v$ for every vertex $v \in U$.
Analogously, $\varphi$ is called a \emph{partial iedm-coloring} if for every such pair of edges $P, Q$ we have that
\begin{equation}
\label{def:partial ledm-col}
\varphi\langle P \setminus Q \rangle \neq \varphi\langle Q \setminus P \rangle
\end{equation}
and $\varphi(v) \in L_v$ for every $v \in U$.
A partial ieds-coloring (iedm-coloring, respectively) is \emph{complete} if $U=V$, note that for $U=V$ the condition \eqref{def:partial leds-col} is exactly \eqref{def:leds-col} and the condition \eqref{def:partial ledm-col} is equivalent to \eqref{def:ledm-col}.
Moreover, notice that conditions \eqref{def:partial leds-col} or \eqref{def:partial ledm-col} are necessary, if we want to extend a partial ieds-coloring or a partial iedm-coloring to a complete one, see \autoref{fig:partial-col}.

\begin{figure}[h]
\begin{center}
\begin{subfigure}[b]{0.5\textwidth}
\centering
\begin{tikzpicture}
\node at (-0.5,2.5) {(a)};
\draw[thick,rounded corners=10pt] (1.5,0)--(0,1.5)--(1.5,3)--(3,1.5)--cycle;
\draw[thick,rounded corners=10pt] (2.5,0)--(1,1.5)--(2.5,3)--(4,1.5)--cycle;
\fill[fill=white, fill opacity=0.9] (3.1,1.73) circle(7pt); 
\filldraw[fill=green!50] (2.5,1.5) circle(4pt);
\draw (1.5,1.5) circle(4pt);
\node at (0.8,1.5) {$v$};
\filldraw[fill=blue!50] (1.5,0.5) circle(4pt);
\filldraw[fill=red!50] (3.5,1.5) circle(4pt);
\filldraw[fill=blue!50] (0.5,1.5) circle(4pt);
\filldraw[fill=red!50] (1.5,2.5) circle(4pt);
\filldraw[fill=blue!50] (2.5,2.5) circle(4pt);
\filldraw[fill=blue!50] (2.5,0.5) circle(4pt);
\end{tikzpicture}
\end{subfigure}%
\begin{subfigure}[b]{0.5\textwidth}
\centering
\begin{tikzpicture}
\node at (-0.5,2.5) {(b)};
\draw[thick,rounded corners=10pt] (1.5,0)--(0,1.5)--(1.5,3)--(3,1.5)--cycle;
\draw[thick,rounded corners=10pt] (2.5,0)--(1,1.5)--(2.5,3)--(4,1.5)--cycle;
\fill[fill=white, fill opacity=0.9] (3.1,1.73) circle(7pt); 
\filldraw[fill=red!50] (2.5,1.5) circle(4pt);
\draw (1.5,1.5) circle(4pt);
\node at (2.2,1.5) {$v$};
\filldraw[fill=blue!50] (1.5,0.5) circle(4pt);
\filldraw[fill=green!50] (3.5,1.5) circle(4pt);
\filldraw[fill=green!50] (0.5,1.5) circle(4pt);
\filldraw[fill=red!50] (1.5,2.5) circle(4pt);
\filldraw[fill=blue!50] (2.5,2.5) circle(4pt);
\filldraw[fill=blue!50] (2.5,0.5) circle(4pt);
\end{tikzpicture}
\end{subfigure}
\end{center}
\caption{Partial colorings which do not satisfy condition \eqref{def:partial leds-col}. After coloring vertex $v$, there is no way to extend these colorings to obtain a complete ieds-coloring. Moreover, (a) does not satisfy condition \eqref{def:partial ledm-col} and cannot be even extended to iedm-coloring.}
\label{fig:partial-col}
\end{figure}
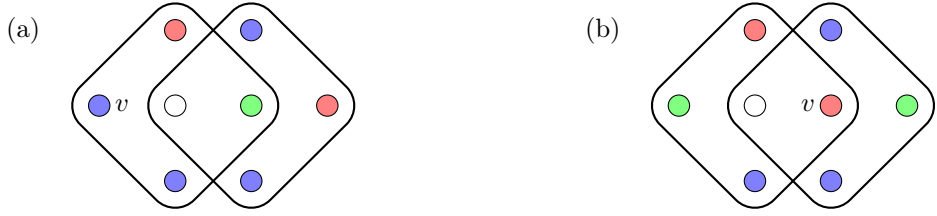

Both statements in \autoref{thm:main} can be shown in a very similar way.
We will discuss the proof in a detail and point out the differences between the cases of sets and multisets.
By a ({\em partial}) \emph{ied-coloring} we will mean a (partial) ieds-coloring or a (partial) iedm-coloring, depending whether we want to show the statement \eqref{thmain-ieds} or \eqref{thmain-iedm}.

We will construct an appropriate coloring of $H$ iteratively, ensuring that after each step the current partial coloring is a partial ied-coloring. In general, the algorithm works as follows.
We fix a large $N \in \mathbb{N}$ and a sequence $C$ of $N$ numbers from $[R]$.
In each step of the algorithm we color an uncolored vertex, using the next number from $C$, and check if the obtained partial coloring is a partial ied-coloring. If so, we proceed to the next iteration.
If not, it means that the condition \eqref{def:partial leds-col} or \eqref{def:partial ledm-col} (depending on the case of sets or multisets) is violated for some pair of edges $P,Q$. In this situation we say that a \emph{conflict} appears on these two edges. We erase colors of some already colored vertices, to ensure that the current coloring is a partial ied-coloring.
Moreover, we use an additional table $T$ to register all information about occurring conflicts. The algorithm terminates when all vertices of $H$ are colored or all numbers from the sequence $C$ are used (i.e., after $N$ iterations).
In the first case the algorithm returns a complete ied-coloring. In the second one, it returns a pair $(T, \varphi_N)$, where $T$ is the table of conflicts and $\varphi_N$ is the partial ied-coloring obtained after the last iteration.

For the contradiction, assume that $H$ does not have any ied-coloring, respecting the lists $L_v$. This means that for every possible sequence $C$ the algorithm does not return a complete list ied-coloring of $H$, but some pair $(T,\varphi_N)$. We will show that there is a bijection between all possible sequences $C$ and all possible pairs $(T,\varphi_N)$. Then, we will   show that if $N$ is sufficiently large, then the number of all pairs $(T,\varphi_N)$ is strictly smaller than $R^N$, which is a number of possible sequences $C$. This leads to a contradiction, so there is at least one sequence $C$, for which the algorithm successfully returns a complete list ied-coloring of $H$.

\subsection{The algorithm} \label{sec:algorithm}
Let $N$ be a large integer and let $C = (c_1, c_2, \dots, c_N)$ be a sequence of integers from $[R]$.
At the beginning of the procedure all vertices of $H$ are uncolored, call this (empty) partial coloring $\varphi_0$. Also, if a list $L_v$ for some $v \in V$ has more than $R$ elements, we truncate it, so that all lists are of size exactly $R$.
For $j \in [N]$, by $\varphi_j$ we denote the partial ied-coloring after the $j$-th iteration.

Recall that the vertices of $H$ are linearly ordered. For each $j \in [N]$, the $j$-th iteration of the algorithm consists of the following steps. 

\begin{description}
\item[\textbf{Step 1.}] Find the smallest uncolored vertex in $V$, call it $v$. 
\item[\textbf{Step 2.}] Assign the $c_j$-th color from the list $L_v$ to $v$ and denote by $\varphi$ the obtained partial coloring.
\item[\textbf{Step 3.}] If $\varphi$ is a partial ied-coloring, write $+$ in $T(j)$ and set $\varphi_j:=\varphi$.
If $\varphi$ is a complete ied-coloring, return it and terminate, otherwise,  proceed to the next iteration.
\item[\textbf{Step 4.}] If $\varphi$ is not a partial ied-coloring, then a conflict appeared on a pair of intersecting edges $P,Q$, and $v$ belongs to at least one of them, say $P$. This is because $\varphi_{j-1}$ was a partial ied-coloring and a conflict was caused by coloring $v$. If there is more than one conflict, we can choose any of them. 
Let $i:=|P \setminus Q|=|Q \setminus P|$, so $i \in I$, and define the set $X := \{K \in E \; \colon \; |P \setminus K|=i \text{ and } v \notin K\}.$

\medskip
The final step differs in case of distinguishing by sets and distinguishing by multisets.
\medskip

\item[\textbf{Step 5 (variant \eqref{thmain-ieds}: sets).}] There are two types of possible conflicts that may occur, we consider them separately.
\begin{enumerate}[a)]
\item First, consider the case that $v \notin Q$. This means that all vertices from $(P \cup Q) \setminus (P \cap Q) \setminus \{v\}$ were colored in previous iterations, and $\varphi(P \setminus Q) \subseteq \varphi(Q)$. 
Let $x_P := n_{E(v)}(P)$ and $x_Q := n_{X}(Q)$. In $T(j)$ write the quadruple $(1,x_P,x_Q,\gamma)$, where $\gamma$ is a color-preserving function from $P \setminus Q$ to $Q$ (later we will specify how to choose $\gamma$, for now it is enough to know that it is color-preserving). After that, uncolor all vertices from $P \setminus Q$, denoting the obtained coloring by $\varphi_j$, and proceed to the next iteration.  
\item Now, consider the case that $v \in P \cap Q$, clearly all vertices of $(P \cup Q) \setminus (P \cap Q)$ are already colored.
Observe that $\varphi(v)$ belongs to exactly one of the sets $\varphi(P \setminus Q)$ and $\varphi(Q \setminus P$), because $\varphi_{j-1}$ was a partial ieds-coloring. Without loss of generality assume that $\varphi(v) \in \varphi(Q \setminus P)$, otherwise switch the names of $P$ and $Q$.

Let $w$ be the smallest vertex of $P \setminus Q$, it exists, because $H$ does not have multiple edges. Let $x_P := n_{E(v)}(P)$ and $x_Q := n_{E(v) \setminus \{P\}}(Q)$.
In $T(j)$ write the quadruple $(2,x_P,x_Q,\gamma)$, where $\gamma$ is a color-preserving function from $P \setminus (Q \cup \{w\}) \cup \{v\}$ to $Q \setminus \{v\}$ . After that, uncolor all vertices from $P \setminus (Q \cup \{w\}) \cup \{v\}$, denoting the obtained coloring by $\varphi_j$, and proceed to the next iteration.  
\end{enumerate}
\item[\textbf{Step 5 (variant \eqref{thmain-iedm}: multisets).}] Note that now $v \in P\setminus Q$, because $\varphi_{j-1}$ was a partial iedm-coloring. Recall that $\varphi \langle P \setminus Q \rangle = \varphi \langle Q \setminus P \rangle$. Let $x_P := n_{E(v)}(P)$ and $x_Q := n_{X}(Q)$. In $T(j)$ write the triple $(x_P,x_Q,\gamma)$, where $\gamma$ is a color-preserving function from $P \setminus Q$ to $Q \setminus P$. After that, uncolor all vertices from $P \setminus Q$, denoting the obtained coloring by $\varphi_j$, and proceed to the next iteration.  
\end{description}

Observe that after uncoloring $v$ (and possibly some additional vertices) there are no more conflicts, so $\varphi_{j}$ is a partial ied-coloring for every $j \in [N]$. Moreover, note that the execution of the algorithm does not depend on the content of $T$. The information stored there will only be used in the proof. 

It is clear that the algorithm either returns an ied-coloring, or runs for $N$ iterations and terminates without success.
In the latter situation, we say that the pair $(T,\varphi_N)$ is \emph{produced} from the sequence $C$, note that this pair is uniquely determined by $C$. Let $\I:=[R]^N$ (for \emph{input}) denote the set of all possible sequences $C$, and $\O$ (for \emph{output}) denote the set of all pairs $(T,\varphi_N)$ that can be possibly produced from a sequence $C \in \I$.
Assume that the algorithm never returns a complete ied-coloring, which means that for every $C \in \I$ it produces some pair $(T,\varphi_N) \in \O$. This implies that $|\mathcal{O}|\leq|\mathcal{I}|$.

\subsection{Equicardinality of sets $\mathcal{I}$ and $\mathcal{O}$}

We will show that if a pair $(T, \varphi_N)$ belongs to $\O$, then there is exactly one sequence $C \in \I$ that produces $(T, \varphi_N)$.
By $W_j$ we denote the set of uncolored vertices after the $j$-th iteration of the algorithm and by $U_j$ we denote the set of vertices that are colored after the $j$-th iteration of the algorithm. Moreover, we define $W_0:=V$ and $U_0:=\phi$.
Notice that for every  $j \in \{0, 1, \dots, N\}$ the sets $W_j$ and $U_j$ form a partition of $V$, i.e., $W_j \cap U_j = \phi$ and $W_j \cup U_j = V$.

First, let us prove that using the entries of $T$ only, we can reconstruct sets $W_j$ for all $j \in [N]$.
\begin{lemma}
For every $j \in [N]$ the set $W_j$ is uniquely determined by $W_{j-1}$ and $T(j)$.
\label{lem:w-reconstr}
\end{lemma}
\begin{proof}
Let $v$ be the smallest vertex in the set $W_{j-1}$. In the $j$-th iteration we assigned a color to $v$. Depending on $T(j)$ and, obviously, the variant of coloring we consider, we have the following possibilities:
\begin{description}
\item[\textbf{Case 1: $T(j)=+$ (both variants \eqref{thmain-ieds} and \eqref{thmain-iedm}: sets and multisets).}] In this case, no conflicts arose after coloring $v$ in the $j$-th iteration and no vertex was uncolored. Therefore $W_j=W_{j-1}\setminus\{v\}.$ 

\item[\textbf{Case 2: $T(j)=(1,x_P,x_Q,\gamma)$ (variant \eqref{thmain-ieds}: sets).}] This means that there was a conflict on some intersecting edges $P$ and $Q$, such that $v \in P \setminus Q$. We know that $P$ is the $x_P$-th element of the set $E(v)$. From the size of the domain of $\gamma$, we can determine the value of $i=|P \setminus Q| = |Q \setminus P|$. Knowing $P$, $v$, and $i$, we can uniquely determine the set $X = \{K \in E \; \colon \; |P \setminus K|=i \text{ and } v \notin K\}$, and thus also the edge $Q$, which is the $x_Q$-th element of $X$.  After the conflict occurred, we uncolored all the vertices from $P \setminus Q$, so $W_j=W_{j-1} \cup (P \setminus Q).$

\item[\textbf{Case 3: $T(j)=(2,x_P,x_Q,\gamma)$ (variant \eqref{thmain-ieds}: sets).}] This means that there was a conflict on some edges $P$ and $Q$, such that $v \in P \cap Q$. We know that $P$ is the $x_P$-th element of $E(v)$, and $Q$ is the $x_Q$-th element of $E(v) \setminus \{P\}$. The vertex $w$ is the smallest vertex in $P \setminus Q$. After the conflict occurred, we uncolored all vertices from $P \setminus (Q \cup \{w\}) \cup \{v\}$, so $W_j=W_{j-1} \cup (P \setminus (Q \cup \{w\})).$ 

\item[\textbf{Case 4: $T(j)=(x_P,x_Q,\gamma)$ (variant \eqref{thmain-iedm}: multisets).}] This case is very similar to the second one. Again we have a conflict on edges $P$ and $Q$, such that $v \in P \setminus Q$. Using $x_P$, we can determine $P$, and $i$ is given by the size of the domain of $\gamma$. Having $v$, $P$, and $i$, we can compute $X$, and then, using $x_Q$, we can find $Q$. After the conflict occurred, we uncolored all the vertices from $P \setminus Q$, so $W_j=W_{j-1} \cup (P \setminus Q).$ \qedhere
\end{description}
\end{proof}

Using \autoref{lem:w-reconstr}, we can reconstruct sets $W_1, \ldots, W_N$. Now we show that using the pair $(T,\varphi_N)$ that was returned by the algorithm and sets $W_j$ for $j \in \{0, \dots, N\}$, we can reconstruct all elements of $C$.

\begin{lemma}
For every $j \in [N]$, the partial ied-coloring $\varphi_{j-1}$ and the number $c_j$ are uniquely determined by $\varphi_j$, $T(j)$, and the sets $W_0, W_1, \dots, W_N$.
\label{lem:c-reconstr}
\end{lemma}
\begin{proof}
For every $j \in \{0, 1, \dots, N\}$ we have $U_j=V\setminus W_j$. Let $v$ be the smallest vertex in $W_{j-1}$, this is the vertex that was colored in the $j$-th iteration. We know that $\varphi_j$ agrees with $\varphi_{j-1}$ on the set $U_j \cap U_{j-1}$. Again we consider the cases.

\begin{description}
\item[\textbf{Case 1: $T(j)=+$ (both variants \eqref{thmain-ieds} and \eqref{thmain-iedm}: sets and multisets).}] This means that coloring $v$ in the $j$-th iteration did not cause any conflict. Thus for every $u \in U_{j-1}$ we have $\varphi_{j-1}(u)=\varphi_j(u)$. From $\varphi_j$ we get $\varphi_j(v)$, and $c_j$ is the position of the color $\varphi_j(v)$ in the list $L_v$.
\item[\textbf{Case 2: $T(j)=(1,x_P,x_Q,\gamma)$ (variant \eqref{thmain-ieds}: sets).}] Recall that this means that there was a conflict on some edges $P$ and $Q$, where $v \in P \setminus Q$. We determine $P$ and $Q$ as in \autoref{lem:w-reconstr}. To get $\varphi_{j-1}$, we need to recover colors of the vertices from $P \setminus (Q \cup \{v\})$, and also the number $c_j$.
Observe that all colors in $\varphi_{j-1}(P \setminus (Q \cup \{v\}))$  appear on vertices of $Q$, and since $\gamma$ is color-preserving, we can easily reconstruct the coloring $\varphi_{j-1}$ on $P \setminus (Q \cup \{v\})$. In the same way we reconstruct the color that was given to $v$ in the $j$-th iteration, $c_j$ is its position in $L_v$.
\item[\textbf{Case 3: $T(j)=(2,x_P,x_Q,\gamma)$ (variant \eqref{thmain-ieds}: sets).}]
This means there was a conflict on some edges $P$ and $Q$, where $v \in P \cap Q$. We find $P$, $Q$, and $w$ as in \autoref{lem:w-reconstr}. As in the case above, we use $\gamma$ to reconstruct colors of $P \setminus (Q \cup \{w\})$ to get $\varphi_{j-1}$, and also the color that was assigned to $v$ to get $c_j$.
\item[\textbf{Case 4: $T(j)=(x_P,x_Q,\gamma)$ (variant \eqref{thmain-iedm}: multisets).}] Again, we find $P$ and $Q$ as in \autoref{lem:w-reconstr}, and then use $\gamma$ to reconstruct $\varphi_{j-1}$ and $c_j$.\qedhere
\end{description}
\end{proof}

Finally, we can use \autoref{lem:w-reconstr} and \autoref{lem:c-reconstr} to get the following corollary.
\begin{corollary}
The sets $\mathcal{I}$ and $\mathcal{O}$ have the same cardinality. 
\label{cor:oi-equal}
\end{corollary}
\begin{proof}
From \autoref{lem:w-reconstr} it follows that knowing $T$,  we can determine the sets $W_j$ for every $j \in [N]$. Using them, the table $T$, and the partial ied-coloring $\varphi_N$,  we can reconstruct all partial ied-colorings $\varphi_j$ and the sequence $C$, as shown in \autoref{lem:c-reconstr}. This implies that every possible pair $(T, \varphi_N) \in \O$ is produced by a unique sequence $C \in \I$, which means that $|\mathcal{I}|\leq |\mathcal{O}|$. Since we know that  $|\mathcal{O}|\leq|\mathcal{I}|$ also holds, the proof is complete.
\end{proof}

\subsection{Cardinality of $\mathcal{O}$}

Now we want to estimate the number of possible pairs $(T,\varphi_N)$ that may be produced by the algorithm.
Let us start with estimating the number of possible entries in the table $T$, other than just a $+$ sign.
We use the notation from the previous section. Consider an iteration $j$, where a conflict appeared on two intersecting edges $P$ and $Q$, such that $v \in P$, and $|P \setminus Q|=|Q \setminus P|=i$, and let $\varphi$ be the coloring obtained by assigning the color $c_j$ to $v$.

We will consider the cases of sets and multisets separately. For each $i \in [k-1]$, by $S_i$ ($M_i$, respectively), we denote the set of all possible entries other than +, that may appear in $T(j)$ in case of ieds-coloring (iedm-coloring, respectively).
Clearly if $i \notin I$, then there are no edges $P$ and $Q$ for which $|P \setminus Q|=i$ holds, so there are no possible conflicts. Therefore in such a case we have $S_i = M_i = \emptyset$. Now consider $i \in I$.

\paragraph*{\textbf{Variant \eqref{thmain-ieds}: sets.}}~Then each entry is a quadruple $(1,x_P,x_Q,\gamma)$ or $(2,x_P,x_Q,\gamma)$.
Since $x_P = n_{E(v)}(P)$, clearly we have $x_P \in [\Delta]$.

\paragraph*{Case 1:  $T(j)=(1,x_P,x_Q,\gamma)$.}~Note that in this case $v \in P \setminus Q$, recall $x_Q = n_X(Q)$, where $X = \{K \in E \; \colon \; |P \setminus K|=i \text{ and } v \notin K\}$. Equivalently, $X$ contains all edges, which do not contain $v$, and have exactly $k-i$ common vertices with $P \setminus \{v\}$. There are at most $(\Delta-1)(k-1)$ edges intersecting $P \setminus \{v\}$, which implies that there are at most $\left\lfloor \frac{(\Delta-1)(k-1)}{k-i}\right\rfloor$ elements of $X$, so \[x_Q \in \left[ \left\lfloor \frac{(\Delta-1)(k-1)}{k-i}\right\rfloor \right].\] 

Finally, we need to estimate the number of possible color-preserving functions $\gamma$, recall $\varphi(P \setminus Q) \subseteq \varphi(Q)$ and we want to store the information about $\varphi(u)$ for every $u \in P \setminus Q$. 
Let $\Gamma_1:=\varphi(P \cap Q)$ and $g_1 :=|\Gamma_1|$, clearly $0 \leq g_1 \leq k-i$.
Also, define $\Gamma_2 := \varphi(Q \setminus P) \setminus \varphi(P \cap Q)$ and $g_2 := |\Gamma_2|$, so we have $0 \leq g_2 \leq i$.
Observe that $\varphi(Q)=\Gamma_1 \cup \Gamma_2$ and $\Gamma_1 \cap \Gamma_2 = \emptyset$. Note that $\Gamma_1$ and $\Gamma_2$ may be considered fixed, as they only depend on $P$, $Q$, and $\varphi$.
Fix a linear ordering of elements of $\varphi(Q)$, which is implied by the linear ordering of vertices of $Q$: each color is represented by the smallest vertex in this color. 

We observe that if the condition \eqref{def:partial leds-col} does not hold, every color from $\Gamma_2$ must appear on some vertex of $P \setminus Q$. On the other hand, colors from $\Gamma_1$ may appear on vertices in $P \setminus Q$, but do not have to. Let $\Gamma_1':= \varphi(P \setminus Q) \cap \Gamma_1$ and $g_1':=|\Gamma'_1|$.
Now $\gamma$ is a surjective function from $P \setminus Q$ to $\Gamma_1' \cup \Gamma_2$ and can be chosen in at most
\[
\sum_{g_1'=0}^{g_1} \binom{g_1}{g_1'} \surj(i, g_1'+g_2) \leq \sum_{g_1'=0}^{k-i} \binom{k-i}{g_1'} \surj(i, g_1'+g_2)
\]
ways, where $\surj(i,g_1'+g_2)$ denotes the number of surjective functions from an $i$-element set to a $(g_1'+g_2)$-element set.
We observe that
\[
\surj(i,g_1'+g_2) = (g_1'+g_2)! \cdot \stirling{i}{g_1'+g_2} \leq f_i,
\]
where the last inequality follows from \eqref{fubini}.
Summing up, we obtain that the number of possible functions $\gamma$ is at most 
\begin{equation}\label{eq: gamma}
\sum_{g_1'=0}^{k-i} \binom{k-i}{g_1'} \surj(i, g_1'+g_2) \leq \sum_{g_1'=0}^{k-i} \binom{k-i}{g_1'} f_i = 2^{k-i} f_i.
\end{equation}

\paragraph*{Case 2:  $T(j)=(2,x_P,x_Q,\gamma)$.}~If $T(j)=(2,x_P,x_Q,\gamma)$, then $v \in P \cap Q$, so $x_Q = n_{E(v) \setminus \{P\}} (Q)$ and thus clearly $x_Q \in [\Delta-1]$. The bound on the number of functions $\gamma$ is obtained in a way analogous to the previous case. Again, we are interested in bounding the number of surjective functions from $[i]$ to $[g_1'+g_2]$, where $g_1'$ is the number of colors from $\varphi(P \cap Q \setminus \{v\})$ that appear in $\varphi(P \setminus (Q \cup \{w\}))$, and $g_2 = |\varphi(Q \setminus P) \setminus \varphi(P \cap Q \setminus \{v\})|$. We observe that this number is bounded by \eqref{eq: gamma}.

Summing up, we conclude that 
\begin{align*}
|S_i| \leq &  \underbrace{\Bigg ( \Delta \cdot \left\lfloor \frac{(\Delta-1)(k-1)}{k-i}\right\rfloor \cdot 2^{k-i} f_i \Bigg)}_{\text{Case 1.}} + \underbrace{\Bigg ( \Delta \cdot (\Delta-1) \cdot 2^{k-i} f_i  \Bigg )}_{\text{Case 2.}} \\
\leq & 2 \Delta \cdot \frac{(\Delta-1)(k-1)}{k-i} \cdot 2^{k-i} f_i = \frac{\Delta(\Delta-1)(k-1)}{k-i} \cdot 2^{k-i+1} f_i.
\end{align*}

\paragraph*{\textbf{Variant \eqref{thmain-iedm}: multisets.}}~This variant is significantly simpler. The only possible entry in $T(j)$ is a triple $(x_P,x_Q,\gamma)$. Just as in the Case 1 in the previous variant, we have 
\[
x_P \in [\Delta]
\qquad\text{ and }\qquad
x_Q \in \left[ \left\lfloor \frac{(\Delta-1)(k-1)}{k-i}\right\rfloor \right].
\]
By condition \eqref{def:partial ledm-col}, $\gamma$ can be assumed to be a bijection, so it can be chosen in $i!$ ways.
So, summing up, we obtain the following bound:
\[
|M_i| \leq \frac{\Delta(\Delta-1)(k-1)}{k-i} \cdot i!.
\]

The rest of the proof is exactly the same in both variants: distingushing by sets and multisets. For every $i \in [k-1]$, by $A_i$ let us denote $S_i$, if we are interested in finding an ieds-coloring, or $M_i$, if we are interested in finding an iedm-coloring. Define $a_i := |A_i|$.
 
Now, let us bound the number of all possible tables $T$ that could be produced by the algorithm, denote it by $\#T$. By $p$ denote a number of $+$ symbols in $T$. For every $i \in [k-1]$ let $t_i$ be the number of appearances of the elements of $A_i$ in $T$; if $i \notin I$, then clearly $t_i=0$.
Notice that $p + t_1+ t_2 +\dots + t_{k-1}= N$. Denote by
\[
{N \choose p,t_1,t_2, \ldots, t_{k-1}}=\frac{N!}{p!t_1! t_2!\ldots t_{k-1}!}
\]
the number of partitions of an $N$-element set into the subsets of cardinalities $p, t_1, t_2, \ldots, t_{k-1}$.

\begin{lemma}
The number $\#T$ of all possible tables $T$ is bounded from above by
\begin{equation}
\sum_{s=N-n+1}^N \; \; \sum_{t_1 + \dots + (k-1)t_{k-1}=s} \ {N \choose p,t_1, \dots, t_{k-1}}a_1^{t_1}\cdot \ldots \cdot a_{k-1}^{t_{k-1}}.
\label{eq: numberoftables}
\end{equation}
\label{lem: numberoftables}
\end{lemma}
\begin{proof}
For every $i \in [k-1]$, the total number of appearances of elements of $A_i$ in $T$ is $t_i$. It implies that for fixed $p,t_1, \ldots, t_{k-1}$, the number of ways to fill the table $T$ is at most
\[
{N \choose p,t_1, \dots, t_{k-1}}a_1^{t_1}\cdot \ldots \cdot a_{k-1}^{t_{k-1}}.
\]
Observe that if an element of $A_i$ appears in $T(j)$, it means that in the $j$-th iteration we uncolored $i-1$ vertices that were colored in previous iterations. Thus to each occurrence of an element of $A_i$ in $T$ we can assign $i-1$ iterations when 
we wrote $+$ into $T$, and each $+$ symbol is assigned at most once.
Let $s$ be the number of iterations when a conflict occurred, or when we colored a vertex that was uncolored later.
There are $t_1+ t_2 + \ldots + t_{k-1}$ iterations when a conflict occured, and $t_2 + 2 t_3 + \ldots + (k-2)t_{k-1}$ iterations when we colored a vertex which was uncolored later, so 
\[
s=(t_1+ t_2 +  \ldots + t_{k-1})+(t_2 + 2 t_3 + \ldots + (k-2)t_{k-1})  =t_1 + 2t_2+ \ldots + (k-1)t_{k-1}.
\]
Clearly $s$ is at most $N$, which is the total number of iterations. On the other hand, after $N$ steps the algorithm returns a partial ied-coloring, where at most $n-1$ vertices are colored. Therefore $N-n+1 \leq s \leq N$.
\end{proof}

We will need the following technical lemma shown in \cite[Lemma 2.7]{BosekAADM}.
\begin{lemma}[Bosek, Czerwiński, Grytczuk, Rz. \cite{BosekAADM}]
\label{lem:estimation}
If $p, t_1, t_2, \dots, t_{k-1}$ are non-negative integers such that $p+t_1+\dots+t_{k-1}=N$ and $t_1+2t_2+\dots+(k-1)t_{k-1}=s \leq N$, then 
\[
{N \choose p,t_1,t_2, \ldots, t_{k-1}} \leq {N \choose s}{s \choose t_1,2t_2,\dots,(k-1)t_{k-1}}q_1^{t_1} \cdot q_2^{2t_2} \cdot \ldots \cdot q_{k-1}^{(k-1)t_{k-1}},
\]
where $q_1=1$ and \[q_i=\frac{i}{i-1}\sqrt[i]{i-1},\]
for $i \geq 1.$
\end{lemma}

Now we are ready to prove the final lemma.
\begin{lemma}
Let
\[
R:=\left\lceil 2+ \sum_{i=1}^{k-1} q_i\sqrt[i]{a_i} \right\rceil,
\]
where $q_1=1$ and $q_i=\frac{i}{i-1}\sqrt[i]{i-1}$ for $i \geq 1.$ Then there exists $N_0$ such that for every $N \geq N_0$ the number of elements of the set $\O$ is strictly smaller than $|\I|=[R]^N$.
\label{lem:oissmaller}
\end{lemma}
\begin{proof}
From Lemma \ref{lem: numberoftables} it follows that the number $\#T$ of all possible tables $T$ that may be produced is bounded by \eqref{eq: numberoftables}. Using \autoref{lem:estimation}, we have that
\begin{align}
\begin{split}
\#T \leq &\sum_{s=N-n+1}^N \; \; \sum_{t_1 + \dots + (k-1)t_{k-1}=s} \ {N \choose p,t_1, \dots, t_{k-1}}a_1^{t_1}\cdot \ldots \cdot a_{k-1}^{t_{k-1}} \\ \leq & \sum_{s=N-n+1}^N {N \choose s} \sum_{t_1 + \dots + (k-1)t_{k-1}=s} {s \choose t_1,2t_2,\dots,(k-1)t_{k-1}} \prod_{i=1}^{k-1} (q_i\sqrt[i]{a_i})^{i t_i}.
\label{eq:numberoftablesestimated}
\end{split}
\end{align}
Because $t_1+2t_2+\dots+(k-1)t_{k-1}=s$, we can use the Multinomial Theorem to write
\begin{equation}
\sum_{t_1+\ldots+(k-1)t_{k-1}=s} {s \choose t_1,2t_2,\dots,(k-1)t_{k-1}} \prod_{i=1}^{k-1} (q_i\sqrt[i]{a_i})^{i t_i} = \left(\sum_{i=1}^{k-1} q_i\sqrt[i]{a_i} \right)^s.
\label{eq:numberoftablesestimated2}
\end{equation}

Now consider the number of all possible partial ied-colorings $\varphi_N$ that can be produced by the algorithm.
Recall that all lists $L_v$ were truncated to exactly $R$ elements.
We can assign one of $R$ colors to each vertex $v \in V$ or leave $v$ uncolored, which gives at most $(R+1)^n$ possible partial ied-colorings. The value $(R+1)^n$ does not depend on $N$, so from \eqref{eq:numberoftablesestimated} and \eqref{eq:numberoftablesestimated2}, if $N$ is sufficiently large, we get
\[
\#T \leq   \sum_{s=1}^N {N \choose s} \left(\sum_{i=1}^{k-1} q_i\sqrt[i]{a_i} \right)^s = \left(1+ \sum_{i=1}^{k-1} q_i\sqrt[i]{a_i} \right)^N.
\]
It means that for $N$ large enough we obtain the following bound on $\O$.
\begin{align}
\begin{split}
|\mathcal{O}| \leq &  \#T \cdot (R+1)^n \leq (R+1)^n \sum_{s=N-n+1}^N{N \choose s} \left(\sum_{i=1}^{k-1} q_i\sqrt[i]{a_i} \right)^s \\
< & (R+1)^n \sum_{s=1}^N {N \choose s} \left(\sum_{i=1}^{k-1} q_i\sqrt[i]{a_i} \right)^s = (R+1)^n \left(1+ \sum_{i=1}^{k-1} q_i\sqrt[i]{a_i} \right)^N \\
\leq & (R+1)^n (R-1)^N  < R^N=|\mathcal{I}|,
\end{split}
\label{eq:r-bound}
\end{align}
which completes the proof of the lemma.
\end{proof}

Assuming that we never produce a complete ied-coloring of $H$, we shown in \autoref{cor:oi-equal} that $|\mathcal{O}|=|\mathcal{I}|$. On the other hand, by \autoref{lem:oissmaller} we get that if $N$ is sufficiently large, then $|\mathcal{O}| <|\mathcal{I}|$, which is a contradiction. Therefore, for any fixed hypergraph $H$ with lists of size at least $R$, there always exists at least one sequence $C$ of length $N$, for which the algorithm returns a complete ied-coloring of $H$.

Looking back at the definition of $a_i$, we get:
\[
R = \begin{cases} 
\left\lceil 2 + \sum_{i \in I} q_i \sqrt[i]{\frac{\Delta(\Delta-1)(k-1)}{k-i}2^{k-i+1} f_i}\right\rceil & \text{in the case of ieds-coloring,}\\[1em]
\left\lceil 2 + \sum_{i \in I} q_i \sqrt[i]{\frac{\Delta(\Delta-1)(k-1)}{k-i}i!}\right\rceil & \text{in the case of iedm-coloring.}
\end{cases}
\]
This completes the proof of both statements in Theorem \ref{thm:main}. \qed

\subsection{Computational complexity}\label{sec:numiters}

In the previous section we showed that if $N$ is sufficiently large, then for some $C$ the algorithm returns a complete ied-coloring. It appears that the algorithm that we used is actually quite efficient. Let us estimate the expected complexity of the algorithm, if $C$ is chosen uniformly at random.

We use the notation from the previous section, in particular $H$ is a fixed hypergraph that we want to color, and $N$ is a large integer. For a sequence $C \in [R]^N$, define a random variable $Z(C)$ as follows:
\[Z(C) = 
\begin{cases}
t & \text{ if the algorithm returns a complete ied-coloring after $t$ iterations,}\\
N+1 & \text{ if the algorithm fails to find a complete ied-coloring.}
\end{cases}
\] 
Note that $Z(C) < N+1$ if and only if the algorithm returns a complete ied-coloring. 
Let us assume that $Z(C) > t$, i.e., the algorithm did not terminate before reaching iteration $t$, and consider the pair $(T, \varphi_t)$ after the $t$-th iteration. Let $\mathcal{O}(t)$ stand for the number of possible pairs $(T, \varphi_t)$ that may be created in the algorithm after $t$ iterations. Observe that if $t > t_0 :=n \; \frac{\ln (R+1)}{\ln \frac{R-0.5}{R-1}}$, from \eqref{eq:r-bound} we get
\[
\mathcal{O}(t) < (R-1)^t(R+1)^n < (R-0.5)^t. 
\]
Hence, for $t \in [t_0+1,N]$, the probability that the algorithm does not terminate after at most $t$ iterations is
\[
\mathbb{P}(Z(C)>t) = \frac{\O(t)}{R^t} < \left (\frac{R-0.5}{R} \right ) ^t.
\]

Assuming that $C$ is chosen uniformly at random, we can estimate the expected value of $Z(C)$ as follows:
\begin{align*}
\mathbb{E} (Z) = &\sum_{t=0}^{N} \mathbb{P}(Z(C)>t) = \sum_{t=0}^{\lceil t_0 \rceil} \mathbb{P}(Z(C)>t)  + \sum_{t= \lceil t_0 \rceil+1}^{N} \mathbb{P}(Z(C)>t) <  \sum_{t=0}^{\lceil t_0 \rceil} 1  +  \sum_{t=\lceil t_0 \rceil+1}^{N} \left (\frac{R-0.5}{R} \right) ^t \\
<& \lceil t_0 \rceil  + 1 + \sum_{t=0}^{\infty} \left (\frac{R-0.5}{R} \right) ^t = \lceil t_0 \rceil  + 1 + \frac{1}{1-\frac{R-0.5}{R}} = \lceil t_0 \rceil + 1 +2R \leq  n \; \frac{\ln (R+1)}{\ln \frac{R-0.5}{R-1}} + 2R +2.
\end{align*}

Notice that 
\[
\lim_{R \to \infty} R \; \ln\left(\frac{R-0.5}{R-1}\right) = \frac{1}{2},
\]
so
\[
\lim_{R \to \infty} \left ( R \ln R \right ) \left ( \ln\left(\frac{R-0.5}{R-1}\right)/\ln(R+1) \right )= \frac{1}{2},
\]
which implies that $\mathbb{E} (Z)=O(nR\ln R$). Moreover, recall that the value of $R$ is bounded by a function of $k$ and $\Delta$, so for fixed $k$ and $\Delta$ the expected number of iterations of the algorithm is linear in $n$.
Finally, observe that each iteration of the algorithm clearly takes polynomial time; actually, if $k$ and $\Delta$ are fixed, then the complexity of each iteration is dominated by finding the least uncolored vertex, which can be done in time $O(\log n)$, using a priority queue.
Summing up, the expected complexity of the algorithm is polynomial, and if $k$ and $\Delta$ are fixed, then it is $O(n \log n)$.

\section{Special cases and extensions}\label{sec:special}
In this section we discuss the applicability of our approach. First, in \autoref{sec:corollaries}, we show how the results in \autoref{thm:main} translate to the graph labeling problems, defined in \autoref{sec:graphlabeling}. We also show another example, where the input hypergraph is a so-called {\em configuration}.
Then, in \autoref{sec:sequences}, we show that our approach can be easily adapted to solve other problems of similar flavor. In particular, we show how to modify the algorithm, so that it produces colorings distinguishing intersecting edges by {\em sequences} of colors.

\subsection{Special cases} \label{sec:corollaries}

Theorem \ref{thm:main} gives us some interesting corollaries concerning graph labeling problems defined in \autoref{sec:graphlabeling}.
Throughout this section, $k$ is a fixed integer greater than 2, and $G=(V,E)$ is a $k$-regular graph, whose every edge $e$  is  assigned with a list $L_{e}$ of available labels. In case of total labelings, we assume that also every vertex $v$ has its list $L_v$.

\begin{corollary} \label{cor:edge-set}
If each of the lists assigned to edges has at least
\begin{equation}
\label{eq:ndes}
\left\lceil 2 + \frac{k-1}{k-2} \sqrt[k-1]{8(k-1)(k-2)f_{k-1}}\right\rceil
\end{equation}
elements, then there exists a list edge labeling of $G$, distinguishing neighbors by sets. Moreover, if $k \geq 1540$, then  \eqref{eq:ndes} is bounded from above by $0.54k$.
\end{corollary}
\begin{proof}
Let $H=(E, Q)$ be the dual hypergraph of $G$. As we observed before, $H$ is $k$-uniform with $\Delta(H)=2$ and $I(H)=\{k-1\}$. From \autoref{thm:main}\eqref{thmain-ieds} we get \eqref{eq:ndes}.

Barth{\'e}l{\'e}my~\cite{barthelemy1980asymptotic} showed that $f_i$ has the following asymptotic behavior:
\[
f_i = \frac{i!}{2(\ln2)^{i+1}}\left(1 + o(1)\right),
\]
so to get the latter bound, observe that 
\[
\lim_{k\rightarrow\infty} \frac{\left\lceil 2 + \frac{k-1}{k-2} \sqrt[k-1]{16(k-1)(k-2)f_{k-1}}\right\rceil}{k}=\lim_{k\rightarrow\infty}\frac{\sqrt[k]{k!\left(1+o(1)\right)}}{k \ln2 }= \frac{1}{e\ln2} \approx 0.531.
\]
It means that if $k$ is sufficiently large, then
\begin{equation} \label{eq:054k}
\frac{1}{e\ln2} \cdot k \leq 0.54k. 
\end{equation}
From a direct calculation we get that \eqref{eq:054k} holds for $k \geq 1540$.
\end{proof}

Analogously we can prove the following three results.

\begin{corollary} \label{cor:total-set}
If each of the lists assigned to edges and vertices has at least
\begin{equation}
\label{eq:ndts}
\left\lceil 2 + \frac{k}{k-1} \sqrt[k]{8k(k-1)f_k}\right\rceil
\end{equation}
elements, then there exists a list total labeling of $G$, distinguishing neighbors by sets. Moreover, if $k \geq 1600$, then  \eqref{eq:ndts} is bounded from above by $0.54k$.
\end{corollary}

\begin{corollary} \label{cor:edge-multiset}
If each of the lists assigned to edges has at least
\begin{equation}
\label{eq:ndem}
\left\lceil 2 + \frac{k-1}{k-2} \sqrt[k-1]{2(k-1)(k-2)(k-1)!}\right\rceil.
\end{equation}
elements, then there exists a list edge labeling of $G$, distinguishing neighbors by multisets. Moreover, if $k \geq 5435$, then  \eqref{eq:ndem} is bounded from above by $0.37k$.
\end{corollary}

\begin{corollary} \label{cor:total-multiset}
If each of the lists assigned to edges and vertices has at least
\begin{equation}
\label{eq:ndtm}
\left\lceil 2 + \frac{k}{k-1} \sqrt[k]{2k(k-1)k!}\right\rceil.
\end{equation}
elements, then there exists a list total labeling of $G$, distinguishing neighbors by multisets. Moreover, if $k \geq 5650$, then  \eqref{eq:ndtm} is bounded from above by $0.37k$.
\end{corollary}

The bounds in Corollaries \ref{cor:edge-multiset} and \ref{cor:total-multiset} follow from the fact that the limit of \eqref{eq:ndem} divided by $k$, and also of \eqref{eq:ndtm} divided by $k$, is $1/e \approx 0.368$.

Let us show another example of application of \autoref{thm:main}. We say that a set of lines is in {\em general position} if no three lines share a common point.

\begin{corollary} \label{cor:lines}
Let $\cal L$ be a family of straight lines in general position.
Fix a family $\cal P$ of points, such that each line in $\cal L$ has precisely $k \geq 3$ points.
Assume that each $p \in {\cal P}$ has a list of at least $R$ colors. For each $p \in {\cal P}$ we can choose a color from its list, so that no two lines that share a common point
\begin{compactenum}[(a)]
\item have the same set of colors, if $R \geq \left\lceil 2 + \frac{k-1}{k-2} \sqrt[k-1]{8(k-1)(k-2)f_{k-1}}\right\rceil$,
\item have the same multiset of colors, if $R \geq \left\lceil 2 + \frac{k-1}{k-2} \sqrt[k-1]{2(k-1)(k-2)(k-1)!}\right\rceil$.
\end{compactenum}
\end{corollary}

\begin{proof}
Construct a hypergraph $H$, whose vertex set is $\cal P$ and edges correspond to $\cal L$. Observe that $H$ is $k$-uniform with $\Delta(H) \leq 2$ and $I=\{k-1\}$. Both claims follow directly from \autoref{thm:main}.
\end{proof}

The hypergraph that appears in \autoref{cor:lines} is a special type of the so-called {\em configuration} (see~\cite{conf}). A $(v,b,k,r)$-configuration is a collection of $v$ points and $b$ lines, where:
\begin{compactenum}[a)]
\item each line contains $k$ points and each point belongs to $r$ lines,
\item two lines share at most one point,
\item each pair of points belongs to at most one line.
\end{compactenum}
In a way analogous to \autoref{cor:lines}, we obtain the following.

\begin{corollary} \label{cor:configuration}
Consider a $(v,b,k,r)$-configuration with set of point $\cal P$ and set of lines $\cal L$, where $k \geq 3$.
Assume that each $p \in {\cal P}$ has a list of at least $R$ colors. For each $p \in {\cal P}$ we can choose a color from its list, so that no two lines that share a common point
\begin{compactenum}[(a)]
\item have the same set of colors, if $R \geq \left\lceil 2 + \frac{k-1}{k-2} \sqrt[k-1]{4r(r-1)(k-1)(k-2)f_{k-1}}\right\rceil$,
\item have the same multiset of colors, if $R \geq \left\lceil 2 + \frac{k-1}{k-2} \sqrt[k-1]{r(r-1)(k-1)(k-2)(k-1)!}\right\rceil$.
\end{compactenum}
\end{corollary}
\begin{figure}[h]
\begin{center}
\begin{subfigure}[c]{0.5\textwidth}
\centering\begin{tikzpicture}
\draw (0,2)--(4,4);
\draw (0,4)--(4,0);
\draw (0,1)--(4,0.5);
\draw (0,1.5)--(4,2);
\draw (0.75,0)--(1,4);
\draw (3,0)--(2.5,4);

\filldraw[fill=yellow!50] (0.8,0.9) circle(4pt);
\filldraw[fill=yellow!50] (2.2,1.8) circle(4pt);
\filldraw[fill=green!50] (2.6,3.3) circle(4pt);
\filldraw[fill=yellow!50] (0.9,2.45) circle(4pt);
\filldraw[fill=red!50] (2.85,1.15) circle(4pt);
\filldraw[fill=yellow!50] (0.85,1.6) circle(4pt);
\filldraw[fill=red!50] (2.9,0.65) circle(4pt);
\filldraw[fill=green!50] (1.35,2.65) circle(4pt);
\filldraw[fill=red!50] (3.5,1.95) circle(4pt);
\filldraw[fill=green!50] (1.9,0.75) circle(4pt);
\end{tikzpicture}
\end{subfigure}%
\begin{subfigure}[c]{0.5\textwidth}
\centering\begin{tikzpicture}
\draw (1,1)--(5,1);
\draw (2,2)--(4,2);
\draw (1,3)--(5,3);
\draw (1,1)--(5,3);
\draw (1,3)--(5,1);
\draw (1,1)--(3,3);
\draw (1,3)--(3,1);
\draw (3,1)--(5,3);
\draw (3,3)--(5,1);

\filldraw[fill=red!50] (1,1) circle(4pt);
\filldraw[fill=red!50] (1,3) circle(4pt);
\filldraw[fill=green!50] (5,1) circle(4pt);
\filldraw[fill=blue!50] (3,1) circle(4pt);
\filldraw[fill=yellow!50] (3,3) circle(4pt);
\filldraw[fill=yellow!50] (5,3) circle(4pt);
\filldraw[fill=blue!50] (2,2) circle(4pt);
\filldraw[fill=green!50] (3,2) circle(4pt);
\filldraw[fill=red!50] (4,2) circle(4pt);
\end{tikzpicture}
\end{subfigure}
\caption{(a) A set of lines in general position and (b) a $(9,9,3,3)$-configuration and their corresponding sets $\cal P$ of points. Observe that in both cases no two lines that share a common point have the same set of colors.}
\end{center}
\end{figure}
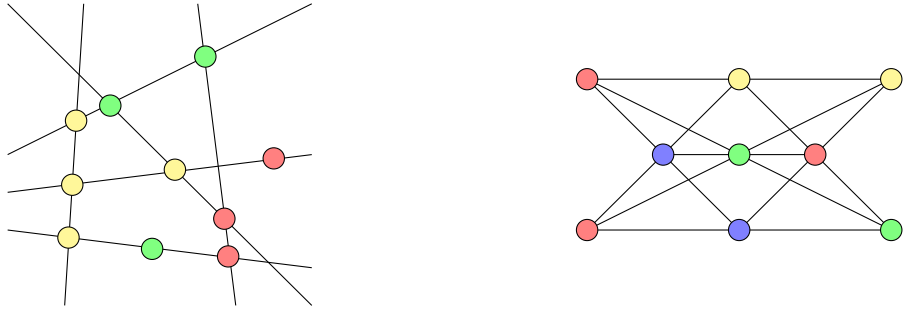

\subsection{Extension: distinguishing intersecting edges by sequences}\label{sec:sequences}

In \autoref{cor:lines} we show how to distinguish intersecting lines by sets and multisets of colors.
However, in case of lines, it is perhaps more natural to consider {\em sequences} of colors. Although \autoref{thm:main} cannot be directly applied in such a case, our approach is robust and can be easily adapted to this setting.
The aim of this section is to present how to adapt the approach from \autoref{sec:entropy} to solve problems of similar flavor. We will not discuss the proof in a detail, as it is almost the same as the proof of \autoref{thm:main}, but we will focus on pointing out the differences.

If edges of our hypergraph correspond to straight lines, each of them has two natural sequences of colors, depending on the direction along the line. We may want to find a coloring, in which the sequences of colors on intersecting edges are different, whichever direction we choose. 
On the other hand, if edges correspond to closed curves, e.g. circles, we may want a coloring in which intersecting edges have distinct ``cyclic sequences'' of colors. Let us start with a generalization of both these cases.

Fix some set $\Pi$ of permutations of $[k]$, such that if $
\sigma \in \Pi$, then $\sigma^{-1} \in \Pi$, and define $\pi := |\Pi|$. We say that two sequences $B=(b_1,b_2,\ldots,b_k)$ and $D=(d_1,d_2,\ldots,d_k)$ are {\em $\Pi$-compatible} if there is $\sigma \in \Pi$, such that $\sigma(B) = (\sigma(b_1),\sigma(b_2),\ldots,\sigma(b_k) )$ is equal to $D$. Note that compatibility is not necessarily reflexive nor transitive.
Now let $H=(V,E)$ be a $k$-uniform hypergraph with maximum degree $\Delta$ and $I(H)=I$.
Let us assume that each edge is ordered, i.e., we will assume that each edge of $H$ is a sequence (such hypergraphs are called {\em sequence hypergraphs} and  were considered e.g. by B\"ohmova, Chalopin, Mihal\'{a}k, Proietti, and Widmayer~\cite{DBLP:conf/wg/BohmovaCMPW16}). Note that there is no relation on sequences of distinct edges, even if they share some vertices.
We say that a coloring of vertices of $H$ {\em distinguishes intersecting edges by $\Pi$-compatible sequences} (or, in short, is {\em $\Pi$-distinguishing}), if for any two intersecting edges $P$ and $Q$, the sequences of colors on $P$ and $Q$ are not $\Pi$-compatible.

Let us show how to adapt our algorithm to find a list $\Pi$-distinguishing coloring of $H$.
We proceed as in the proof of \autoref{thm:main}. Analogously to \eqref{def:partial leds-col} and \eqref{def:partial ledm-col}, by a conflict we mean a situation, that the current partial coloring cannot be extended to a complete $\Pi$-distinguishing coloring.
More formally, consider two $k$-element sequences $P'=(p'_1,p'_2,\ldots,p'_k)$ and $Q=(q_1,q_2,\ldots,q_k)$ of vertices, sharing at least one common element (we will interpret $P'$ as some permutation of an edge $P$, and $Q$ is just an edge). Let $\varphi$ be a partial coloring of $P' \cup Q$ (for simplicity, we will sometimes identify sequences with their sets of elements). For $j \in [k]$, we say that $P'$ is {\em $\varphi$-similar to $Q$ on position $j$} if either $p'_j$ and $q_j$ are colored by $\varphi$ and $\varphi(p'_j)=\varphi(q_j)$, or $p'_j = q_j$, i.e., the same vertex appears in both sequences on position $j$.
The edges $P$ and $Q$ are {\em $\varphi$-similar} if there is $\sigma \in \Pi$, such that for every $j \in [k]$, sequences $\sigma(P)$ and $Q$ are $\varphi$-similar on position $j$.

We observe that if for any two intersecting edges $P$ and $Q$, and any partial vertex coloring $\varphi$, edges $P$ and $Q$ are $\varphi$-similar, there is no way to extend $\varphi$ to a complete $\Pi$-distinguishing coloring of $H$. So in this situation we will report a conflict.
Note that, similarly to the case of sets, a conflict on intersecting edges $P$ and $Q$ may occur after coloring a vertex $v$ from $P \setminus Q$ or from $P \cap Q$ (see \autoref{fig:lines-conflict}).
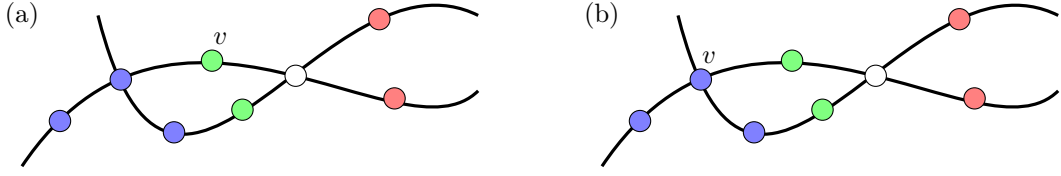
\begin{figure}[h]
\begin{center}
\begin{subfigure}[c]{0.5\textwidth}
\begin{tikzpicture}
\node at (0,2) {(a)};
\draw[very thick] (0,0) .. controls (2,3) and (5,0) .. (6,1);
\draw[very thick] (1,2) .. controls (2,-2) and (4,3) .. (6,2);

\filldraw[fill=blue!50] (1.3,1.15) circle(4pt);
\filldraw[fill=white!50] (3.6,1.2) circle(4pt);
\filldraw[fill=blue!50] (0.5,0.6) circle(4pt);
\filldraw[fill=green!50] (2.5,1.4) circle(4pt);
\node at (2.55,1.7) {$\ v$};
\filldraw[fill=blue!50] (2,0.45) circle(4pt);
\filldraw[fill=green!50] (2.9,0.75) circle(4pt);
\filldraw[fill=red!50] (4.9,0.9) circle(4pt);
\filldraw[fill=red!50] (4.7,1.95) circle(4pt);
\end{tikzpicture}
\end{subfigure}%
\begin{subfigure}[c]{0.5\textwidth}
\begin{tikzpicture}
\node at (0,2) {(b)};
\draw[very thick] (0,0) .. controls (2,3) and (5,0) .. (6,1);
\draw[very thick] (1,2) .. controls (2,-2) and (4,3) .. (6,2);

\filldraw[fill=blue!50] (1.3,1.15) circle(4pt);
\node at (1.35,1.45) {$\ v$};
\filldraw[fill=white!50] (3.6,1.2) circle(4pt);
\filldraw[fill=blue!50] (0.5,0.6) circle(4pt);
\filldraw[fill=green!50] (2.5,1.4) circle(4pt);
\filldraw[fill=blue!50] (2,0.45) circle(4pt);
\filldraw[fill=green!50] (2.9,0.75) circle(4pt);
\filldraw[fill=red!50] (4.9,0.9) circle(4pt);
\filldraw[fill=red!50] (4.7,1.95) circle(4pt);
\end{tikzpicture}
\end{subfigure}
\caption{Edges $P$ and $Q$ of the hypergraph and their corresponding sequences, represented by ($x$-monotone) curves. After coloring a vertex $v$ such that (a) $v \in P \setminus Q$  or (b) $v \in P \cap Q$, it is not possible to extend the obtained partial coloring to get different sequences on $P$ and $Q$. Note that in both cases there was no conflict before coloring $v$.}
\label{fig:lines-conflict}
\end{center}
\end{figure}

Suppose a conflict given by some $\sigma \in \Pi$ appeared on intersecting edges $P$ and $Q$, such that $v \in P \setminus Q$ and $|P \setminus Q| = |Q \setminus P| = i$. In this case we want to erase all colors from $P \setminus Q$.
In our encoding of this conflict, we need to be able to retrieve the edges $P$ and $Q$. Moreover, to remember the coloring of vertices in $P \setminus Q$, we need to remember the permutation $\sigma$. Since $\sigma \in \Pi$, it can be chosen in at most $\pi$ ways. Observe that we do not erase the colors of any vertices of $Q$, so this is indeed sufficient to restore the coloring.

We can choose one of two possible encodings of $P$ and $Q$. 
In the first one, we store $x_P$ and $x_Q$, in the same way as in the case of sets. Additionally, in order to retrieve the set $X$ and thus the edge $Q$, we need to be able to store the information about $i$. We encode it as $n_I(i)$.
The total number of possibilities is at most $\Delta \left\lfloor \frac{(\Delta-1)(k-1)}{k-i}\right\rfloor |I| \; \pi$.

In the second variant, we again store $x_P$ as $n_{E(v)}(P)$. Then we choose any vertex $v'$ in $P \cap Q$ (note that we can do it in at most $k-1$ ways, as $v \notin Q$), and then encode $Q$ as $n_{E(v') \setminus \{P\}}(Q)$. 
The total number of possibilities is at most $\Delta (\Delta-1)(k-1) \; \pi$.

For every $i \in I$ we choose the encoding that gives fewer possibilities (and store an information about the choice as well), so in the case if $v \in P \setminus Q$, the number of possibilities is at most:
\[
\min \left \{\Delta \left\lfloor \frac{(\Delta-1)(k-1)}{k-i} \right\rfloor |I| \; \pi,\; \Delta (\Delta-1)(k-1) \; \pi \right \} \leq \Delta(\Delta-1)(k-1) \pi \min \left\{\frac{|I|}{k-i},1 \right\}.
\]

If there is a conflict on edges $P$ and $Q$, such that $v \in P \cap Q$, we also proceed as in the case of sets. We encode $P$ and $Q$ with $x_P=n_{E(v)}(P)$ and $x_Q=n_{E(v) \setminus \{P\}}(Q)$. We also store $\sigma$.
Let $v'$ be the vertex of $P$, which is mapped to $v$ by $\sigma$ applied to $P$ (recall that $v$ is a vertex of $Q$). If $v' \in P \setminus Q$, then let $w:=v'$, otherwise let $w$ be the smallest vertex of $P \setminus Q$. We will erase the colors of vertices of $P \setminus (Q \cup \{w\}) \cup \{v\}$. Note that by the choice of $w$, we can safely do it, as we can restore the coloring using $\sigma$.

Also note that we always erase the color of $v$, so after that all conflicts are removed.
We observe that the maximum number of possible entries in the conflict table $T$, which correspond to a conflict on $P$ and $Q$, is

\[
a_i \leq \Delta(\Delta-1)(k-1) \pi \min \left\{1, \frac{|I|}{k-i} \right\} + \Delta(\Delta-1)\pi = \Delta(\Delta-1)\pi  \left(1+ (k-1)\min \left\{1, \frac{|I|}{k-i} \right\}\right).
\]

The rest of the proof is exactly the same, but let us add a slight modification in the very end. While proving \autoref{thm:gndi-nphard}, we did not care too much about the additive constant in the bound on the sizes of lists, but now let us estimate \eqref{eq:r-bound} more carefully. 
Now, for every $\epsilon >0$, if $N$ is large enough,  we have
\begin{align*}
|\mathcal{O}| \leq & \#T \cdot (R+1)^n  \leq (R+1)^n  \sum_{s=N-n+1}^N{N \choose s} \left(\sum_{i=1}^{k-1} q_i\sqrt[i]{a_i} \right)^s \\
 \leq &\left(\sum_{i=1}^{k-1} q_i\sqrt[i]{a_i} \right)^N(R+1)^{n} \sum_{s=N-n+1}^N{N \choose s} \left(\sum_{i=1}^{k-1} q_i\sqrt[i]{a_i} \right)^{s-N} \\
 \leq &\left(\sum_{i=1}^{k-1} q_i\sqrt[i]{a_i} \right)^N(R+1)^{n} \cdot n \cdot N^n 
 <  \left(\epsilon/2 + \sum_{i=1}^{k-1} q_i\sqrt[i]{a_i}\right)^N (R+1)^n < \left(\epsilon + \sum_{i=1}^{k-1} q_i\sqrt[i]{a_i}\right)^N.
\end{align*}
So, choosing $\epsilon$ sufficiently small, if \[R \geq \left \lceil \epsilon + \sum_{i=1}^{k-1} q_i\sqrt[i]{a_i} \right \rceil= 1+\left\lfloor \sum_{i=1}^{k-1} q_i\sqrt[i]{a_i} \right \rfloor,\] then $|\O|<|\I|$, which completes the proof. Thus we obtain the following.

\begin{theorem}\label{thm:seq}
Let $k \geq 2$ and let $H$ be a $k$-uniform hypergraph $H$ with $\Delta(H) \leq \Delta$ and $I(H)=I$, whose every vertex is equipped with a list of at least $R$ colors. Let $\Pi$ be some set of permutations of $[k]$, such that if $
\sigma \in \Pi$, then $\sigma^{-1} \in \Pi$. Let $\pi := |\Pi|$. Define $q_1:=1$ and $q_i:= \frac{i}{i-1}\sqrt[i]{i-1}$ for every $i > 1$.
If $R \geq 1+\left\lfloor \sum_{i \in I} q_i \sqrt[i]{\Delta(\Delta-1)\pi  \left(1+ (k-1)\min \left\{1, \frac{|I|}{k-i} \right\}\right)}\right\rfloor$,
then there exists a list $\Pi$-distinguishing coloring of $H$.
\end{theorem}

So, going back to the example of lines, we obtain the following sequence counterpart of \autoref{cor:lines}, see \autoref{fig:lines-seq}.
\begin{corollary} \label{cor:lines-seq}
Let $\cal L$ be a family of straight lines in a general position.
Fix a family $\cal P$ of points, such that each line in $\cal L$ has precisely $k \geq 3$ points.
Assume that each $p \in {\cal P}$ has a list of at least $R$ colors (regardless on the direction along the line). 
If \[R \geq 1+\left\lfloor \frac{k-1}{k-2} \sqrt[k-1]{4k^2-8k}\right\rfloor,\]
then for each $p \in {\cal P}$ we can choose a color from its list, so that no two lines that share a common point have the same sequence of colors.
\end{corollary}
Note that for $k \geq 12$, the bound in \autoref{cor:lines-seq} is 2, which is clearly best possible.

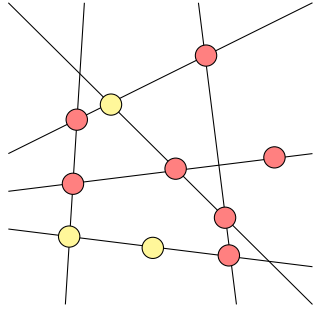
\begin{figure}[h]
\begin{center}
\begin{tikzpicture}
\draw (0,2)--(4,4);
\draw (0,4)--(4,0);
\draw (0,1)--(4,0.5);
\draw (0,1.5)--(4,2);
\draw (0.75,0)--(1,4);
\draw (3,0)--(2.5,4);

\filldraw[fill=yellow!50] (0.8,0.9) circle(4pt);
\filldraw[fill=red!50] (2.2,1.8) circle(4pt);
\filldraw[fill=red!50] (2.6,3.3) circle(4pt);
\filldraw[fill=red!50] (0.9,2.45) circle(4pt);
\filldraw[fill=red!50] (2.85,1.15) circle(4pt);
\filldraw[fill=red!50] (0.85,1.6) circle(4pt);
\filldraw[fill=red!50] (2.9,0.65) circle(4pt);
\filldraw[fill=yellow!50] (1.35,2.65) circle(4pt);
\filldraw[fill=red!50] (3.5,1.95) circle(4pt);
\filldraw[fill=yellow!50] (1.9,0.75) circle(4pt);
\end{tikzpicture}
\end{center}
\caption{A set of lines in general position and a set $\cal P$ of points. No two lines that share a common point have the same set of colors, regardless of directions.} \label{fig:lines-seq}
\end{figure}

Finally, let us recall that Seamone and Stevens~\cite{DBLP:journals/dmtcs/SeamoneS13} considered distinguishing adjacent vertices of a graph by sequences of labels appearing on the incident edges. In their setting, there was one global ordering of all edges (either fixed, or one that could be chosen along with the labeling), and sequences of labels were implied by this ordering. For fixed ordering of edges, they showed that one can choose such a labeling from lists of size at least 3, provided that the minimum degree is large enough, compared to the maximum degree. In particular, lists of size 3 suffice for $k$-regular graphs with $k \geq 6$.

Note that if $\Pi$ contains only the identity, then a bound for the size of lists in a $\Pi$-distinguishing coloring of the dual hypergraph of a $k$-regular graph implies the bound for the size of lists in the problem considered by Seamone and Stevens. Actually, our setting is more flexible, as the sequences corresponding to edges of the hypergraph can be chosen independently.
However, even in this more general setting, we improve the result by Seamone and Stevens~\cite{DBLP:journals/dmtcs/SeamoneS13} for regular graphs.

\begin{corollary} \label{cor:graph-sequences}
Let $G$ be a $k$-regular graph, and fix some ordering of edges.
If each edge has a list of at least
\begin{equation}
\label{eq:graph-seq}
1+\left\lfloor \frac{k-1}{k-2} \sqrt[k-1]{2k^2-4k}\right\rfloor
\end{equation}
labels, then there exists a list edge labeling, in which no two adjacent vertices have the same sequence of labels appearing on incident edges. Moreover, if $k \geq 6$, then \eqref{eq:graph-seq} is at most 3, and if $k \geq 10$, then  \eqref{eq:graph-seq} is $2$.
\end{corollary}

Finally, note that in the reasoning in this section, we could in fact have a different set $\Pi$ for each (ordered) pair of intersecting edges. Also, our algorithm can be easily adapted to different problems of similar flavor.

\section{Generalized neighbor distinguishing index of bipartite graphs} \label{sec:gndi}
Gy\H{o}ri , Hor\v{n}\'ak, Palmer, and Woźniak \cite{GYORI2008827} considered the generalized neighbor-distinguishing index of bipartite graphs. They showed that for bipartite $G$ with at least 3 vertices it holds that $\gndi(G) \in \{2,3\}$, and provided some observations on graphs with $\gndi(G)=2$.
Let us summarize them below. The notation from \autoref{obs:gndi2} will be used throughout this section.

\begin{observation}[\cite{GYORI2008827}]\label{obs:gndi2}
Let $G=(V,E)$ be a connected graph with $\gndi(G)=2$, and consider an edge labeling $c \colon E \to \{1,2\}$, which distinguishes neighbors by sets. Then the following hold:
\begin{compactenum}
\item $G$ is bipartite,
\item the vertices $v \in V$, such that $c(E(v))=\{1,2\}$ form one bipartition class $Y$, and the remaining vertices form the second bipartition class $X$,
\item $X$ is partitioned into two subsets: $X_1$ containing vertices $v$, such that $c(E(v)) = \{1\}$, and $X_2$ containing vertices $v$, such that $c(E(v)) = \{2\}$,
\item if $\deg v = 1$, then $v \in X$,
\item if $v \in Y$ and $\deg v =2$, then one neighbor of $v$ is in $X_1$ and the other is in $X_2$.
\end{compactenum}
\end{observation}

In this section we investigate the edge labelings distinguishing neighbors by sets in bipartite graphs more closely. In particular, we exhibit an interesting relation between the problem of determining $\gndi(G)$ for bipartite $G$, and a well-known variant of satisfiability problem, called \naesat, which can be equivalently formulated as a hypergraph coloring problem. This relation was also observed by  Gy{\H{o}}ri and Palmer \cite{DBLP:journals/dm/GyoriP09}.

\subsection{Complexity aspects}

The instance $\Phi$ of \naesat (usually called a {\emph formula}) is a set of boolean variables and a collection of clauses, each being a set of literals (negated or non-negated variables). A clause is unsatisfied if all its literals have the same value, otherwise it is satisfied. The problem asks whether there exists a truth assignment in which all clauses are satisfied.
A restriction of \naesat, in which each clause has at most $k$ literals, is called \naeksat{$k$}. Finally, in \posnaesat, there are no negated variables.
It is well known that the \posnaeksat{3} is NP-complete~\cite{Lovasz,Schaefer:1978:CSP:800133.804350}. Moreover, an algorithm solving \posnaeksat{3} in time $2^{o(N + M)}$, where $N$ is the number of variables and $M$ is the number of clauses, would contradict the Exponential Time Hypothesis (ETH), which is one of the standard assumptions in complexity theory~\cite{DBLP:books/sp/CyganFKLMPPS15}.

For an instance $\Phi$ of \naesat with variables $\cV$ and clauses $\cC$, its \emph{incidence graph} is the bipartite graph with bipartition classes $\cV$ and $\cC$, in which $v \in \cV$ and $C \in \cC$ are adjacent if an only if $v \in C$.
The {\em literal-clause incidence graph} of a formula $\Phi$ is similar to the incidence graph, but now vertices correspond to \emph{literals} and clauses, and edges indicate whether a literal belongs to the clause. Moreover, we have additional edges between each literal and its negation. 
By \planaesat we mean the restriction of \naesat, where the literal-clause incidence graph is planar. Note that in case of \posnaesat, there are no negated variables, so literal-clause incidence graph is exactly the incidence graph, and thus strong planarity of $\Phi$ is equivalent to the planarity of its incidence graph.

Analogously, for a bipartite graph $G$ with bipartition classes $A$ and $B$, an \emph{$A$-derived formula} is a \posnaesat formula, in which for every $v \in A$ we have a variable $var(v)$, and for every $u \in B$ we have a clause $C_u = \{ var(v) \colon v \in N(u) \}$. Observe that the incidence graph of an $A$-derived formula of $G$ is $G$.

\begin{lemma} \label{lem:gndi-naesat}
Let $G$ be a connected bipartite graph with bipartition classes $A,B$ and let $\Phi$ be its $A$-derived formula.
Then $\Phi$ is satisfiable if and only if $G$ has an edge labeling $c$ with two labels, distinguishing neighbors by sets, in which $A=X$.
\end{lemma}
\begin{proof}
First, suppose $\Phi$ is satisfiable and let $\varphi$ be a satisfying truth assignment.
We define an edge labeling $c$ as follows: for $v \in A$ and $u \in B$, we set $c(uv) = 1$ if $\varphi(var(v))=true$ and $c(uv) = 2$ if $\varphi(var(v))=false$. Note that $c$ is well-defined and for each $v \in A$ we have $c(E(v))=\{\varphi(var(v))\}$. Now consider a clause $C_u$. Since $\varphi$ satisfies $C_u$, there is some $var(v_1) \in C_u$ such that $\varphi(var(v_1))=true$ and some $var(v_2) \in C_u$ such that $\varphi(var(v_2))=false$. This means that $c(v_1u) = 1$ and $c(v_2u)=2$, so $c(E(u))=\{1,2\}$. Therefore $c$ satisfies the conditions in the statement of the lemma.

On the other hand, for an edge labeling $c \colon E \to \{1,2\}$, which distinguishes neighbors by sets and in which $A=X$, we may define the truth assignment $\varphi$ as $\varphi(var(v))=true$ if $v \in X_1$ and $\varphi(var(v))=false$ if $v \in X_2$. Since for each $u$ we have $c(E(u))=\{1,2\}$, we observe that $C_u$ contains and least one true and at least one false variable.
\end{proof}

\autoref{lem:gndi-naesat} together with NP-completeness of \posnaesat implies that it is NP-complete to decide whether an input bipartite graph $G$ with bipartition classes $A,B$ has an edge labeling with two labels, which distinguishes neighbors by sets and $A=X$. Here we strengthen this result by restricting the structure of $G$ and removing asymmetry between bipartition classes.

\begin{theorem}  \label{thm:gndi-nphard}
For every $g \geq 4$, given a subcubic bipartite graph $G$ with $n$ vertices and girth at least $g$, it is NP-complete to decide whether $\gndi(G)=2$. Moreover, the problem cannot be solved in time $2^{o(n)}$, unless the ETH fails.
\end{theorem}

\begin{proof}
Consider an instance $\Phi$ of \posnaeksat{3}, with variable set $\cV$ and clause set $\cC$. Let $N:=|\cV|$ and $M:=|\cC|$.
We may assume that for every $C \in \cC$ we have $2 \leq |C| \leq 3$, since one-element clauses cannot be satisfied.
For a variable $v$, by $d(v)$ we will denote the total number of occurrences of $v$ in clauses. Clearly $\sum_{v \in \cV} d(v) \leq 3M$. 
We assume that the clauses are ordered, this implies the ordering of occurrences of each variable.

Let $g'$ be the smallest integer divisible by 4, which is larger than $g/2$.
For each variable $v \in \cV$ we introduce a path $P^v$ with $\ell=g' \cdot d(v) +1$ vertices $p^v_0,p^v_1,\ldots,p^v_{\ell-1}$.
For $i=1,2,\ldots,d(v)$, the vertex $p^v_{(i-1) \cdot g'}$ will be denoted by $x^v_i$, and the vertex $p^v_{\ell-1}$ will be denoted by $y^v$.

For each clause $C$ we introduce a vertex $z^C$. Now, if $C$ contains variables $u,v,w$, such that the occurrence in $C$ is the $i$-th occurrence of $u$, the $j$-th occurrence of $v$, and the $k$-th occurrence of $w$, we add edges $x^u_iz^C$, $x^v_jz^C$, and $x^w_kz^C$.
If $|C|=2$, then we add just two edges.

Let $G$ be the graph constructed in such a way. Observe that $G$ is bipartite, subcubic, and the girth of $G$ is at least $2g' \geq g$ (see \autoref{fig:gndi-girth}).
Moreover, for every $v \in \cV$, the degree of $y^v$ is 1.

\begin{figure}[ht]
\centering
\begin{tikzpicture}[scale=0.6,every node/.style={circle,minimum size = 5pt},red/.style={fill=red!20}, blue/.style={fill=blue!20},
, rb/.style  = {path picture={
    \fill[blue!20](path picture bounding box.north west) -- (path picture bounding box.south east) --  (path picture bounding box.north east)--cycle;
    \fill[red!20](path picture bounding box.north west) -- (path picture bounding box.south west) --  (path picture bounding box.south east)--cycle;
  }}]

\foreach \i in {0,1}
{	
	\node[draw] (u1\i) at (12*\i,5) {};
	\node[draw] (u2\i) at (12*\i+1,5) {};
	\node[draw] (u3\i) at (12*\i+2,5) {};
	\node[draw] (u4\i) at (12*\i+3,5) {};
	
	\node[draw] (v1\i) at (12*\i+5,5) {};
	\node[draw] (v2\i) at (12*\i+6,5) {};
	\node[draw] (v3\i) at (12*\i+7,5) {};
	\node[draw] (v4\i) at (12*\i+8,5) {};
	
	\draw (u1\i)--(u2\i);
	\draw[line width=2, color=red!50] (u2\i)--(u3\i)--(u4\i);
	
	\draw[line width=2, color=red!50] (v1\i)--(v2\i);
	\draw (v2\i)--(v3\i)--(v4\i);
	
	\draw[dashed,line width=2, color=red!50] (u4\i) -- (v1\i);
	\draw[dashed] (u1\i) --++ (-1,0);
	\draw[dashed] (v4\i) --++ (1,0);
	
	\draw [
    thick,
    decoration={
        brace,        
        raise=0.2cm
    },
    decorate
] (u2\i) -- (v2\i) 
node [pos=0.5,anchor=north,yshift=0.95cm] {$\geq g'$}; 
}

\node[label={\small $x^u_i$}]  at  (u20) {};
\node[label={\small $x^u_{i+1}$}] at  (v20) {};
\node[label={\small $x^v_j$}]  at  (u21) {};
\node[label={\small $x^v_{j+1}$}]  at  (v21) {};

\node[draw,label=below:{\small $z^C$}] (zc1) at (8,0) {};
\node[draw,label=below:{\small $z^{C'}$}] (zc2) at (12,0) {};

\draw[line width=2, color=red!50] (u20) -- (zc1)--(u21);
\draw[line width=2, color=red!50] (v20) -- (zc2)--(v21);

\end{tikzpicture}
\caption{The girth of $G$ is at least $2g' \geq g$.} \label{fig:gndi-girth}
\end{figure}
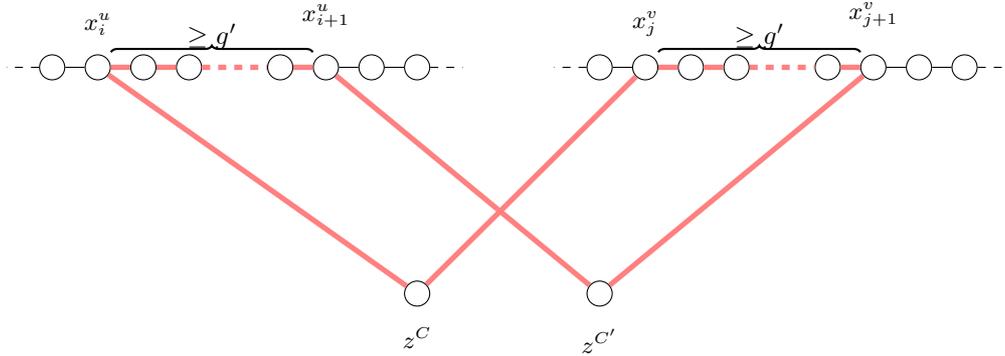

We claim that $\gndi(G)=2$ if and only if $\Phi$ is satisfiable.
First, suppose that $\gndi(G)=2$ and let $c \colon E(G) \to \{1,2\}$ be a labeling distinguishing neighbors by sets.
Recall from \autoref{obs:gndi2} that for every $v \in \cV$, the value of $c(E(y^v))$ is either $\{1\}$ or $\{2\}$.
Moreover, it is straightforward to observe that for all $v$ and all $i$ it holds that $c(E(x^v_i))=c(E(y^v))$ (see \autoref{fig:gndi-var}).
Let $\varphi$ be the truth assignment defined as follows: $\varphi(v)$ is $true$ if $c(E(y^v))=\{1\}$, otherwise $\varphi(v)$ is $false$.
Suppose $\varphi$ is not a satisfiable assignment and let $C$ be an unsatisfied clause. We consider the case that $|C|=3$, the case if $|C|=2$ is analogous. So let $C=\{u,v,w\}$ and without loss of generality assume that all these three variables are set false by $\varphi$.
However, this means that all edges incident with $z^C$ are labeled 2, so $c(E(z^C))=\{2\}$, which contradicts the assumption that $c$ distinguishes neighbors by sets.

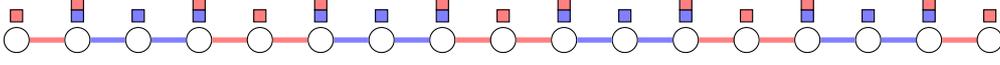
\begin{figure}[h]
\centering
\begin{tikzpicture}[scale=0.8,every node/.style={draw,circle,minimum size = 5pt},white/.style={fill=white}]

\foreach \i in {0,...,3}
{
    \draw [fill=red!50] (-0.1+4*\i,0.3) rectangle (0.1+4*\i,0.5);
	\node[white] (x\i) at (4*\i,0) {};
	\draw [fill=blue!50] (0.9+4*\i,0.3) rectangle (1.1+4*\i,0.5);
	\draw [fill=red!50] (0.9+4*\i,0.5) rectangle (1.1+4*\i,0.7);
	\node[white]  (a\i) at (4*\i+1,0) {};
	\draw [fill=blue!50] (1.9+4*\i,0.3) rectangle (2.1+4*\i,0.5);
	\node[white] (y\i) at (4*\i+2,0) {};
	\draw [fill=blue!50] (2.9+4*\i,0.3) rectangle (3.1+4*\i,0.5);
	\draw [fill=red!50] (2.9+4*\i,0.7) rectangle (3.1+4*\i,0.5);
	\node[white]  (b\i) at (4*\i+3,0) {};
	\draw[line width=2, color=red!50] (x\i) -- (a\i);
	\draw[line width=2, color=blue!50] (a\i)--(y\i) -- (b\i);
	\draw[line width=2, color=red!50] (b\i)--(4*\i+4,0);
}
\draw [fill=red!50] (-0.1+4*4,0.3) rectangle (0.1+4*4,0.5);
\node[white] (x4) at (4*4,0) {};
\end{tikzpicture}
\caption{Up to the permutation of labels, there is only one possible labeling of edges of $P^v$.} \label{fig:gndi-var}
\end{figure}

On the other hand, let $\varphi$ be a satisfiable assignment of $\Phi$. We define the labeling $c$ of edges of each $P^v$ according to the pattern shown in \autoref{fig:gndi-var}, such that the edge incident with $y^v$ is labeled 1 if $\varphi(v)=true$, and otherwise it is labeled 2.
The edges incident to $x^v_i$ are labeled in such a way that $c(E(x^v_i))$ is either $\{1\}$ or $\{2\}$, note that this is uniquely determined by the labeling of edges of $P^v$. Note that since $\varphi$ satisfied $\Phi$, we have $c(E(z^C))=\{1,2\}$ for every $C \in \cC$, and thus $c$ distinguishes neighbors by sets.

Finally, observe that the number $n$ of vertices of $G$ is $O(N+M)$, so an algorithm solving our problem in time $2^{o(n)}$ could be used to solve any instance of \posnaeksat{3} in time $2^{o(N+M)}$, which would in turn contradict the ETH.
\end{proof}

\autoref{lem:gndi-naesat} has some further interesting implications. 
Quite surprisingly, Moret~\cite{Moret:1988:PNP:49097.49099} showed that \planaeksat{3} can be solved in polynomial time.
Since the textbook polynomial reduction from \naesat to \naeksat{3} preserves planarity of the incidence graph, the following also holds.

\begin{corollary}\label{cor:planae}
\planaesat can be solved in polynomial time.
\end{corollary}

This,  combined with \autoref{lem:gndi-naesat}, gives us the following result.

\begin{theorem}\label{thm:gndi-planar-poly}
If $G$ is planar and bipartite, then $\gndi(G)$ can be computed in polynomial time.
\end{theorem}
\begin{proof}
Recall that if $G$ is bipartite, then either $\gndi(G)=2$ or $\gndi(G)=3$.  We will focus on deciding whether $\gndi(G)=2$, otherwise we answer that $\gndi(G)=3$.
Without loss of generality we may assume that $G$ is connected, since otherwise we may label each connected component independently. Let $A,B$ be bipartition classes of $G$.
Clearly $\gndi(G)=2$ if and only if at least one of the following hold:
\begin{compactitem}
\item $G$ has an edge labeling $c$ with two labels, distinguishing neighbors by sets, in which $A=X$, or
\item $G$ has an edge labeling $c$ with two labels, distinguishing neighbors by sets, in which $B=X$.
\end{compactitem}
By \autoref{lem:gndi-naesat}, this is equivalent to saying that the $A$-derived formula $\Phi_A$ of $G$ is satisfiable, or the $B$-derived formula $\Phi_B$ of $G$ is satisfiable. Notice that the incidence graph of $\Phi_A$ and of $\Phi_B$ is exactly $G$, so it is planar. Since formulae $\Phi_A$ and $\Phi_B$ can be computed in polynomial time, and by \autoref{cor:planae}, the satisfiability of each of them can be tested in polynomial time, we obtain our claim.
\end{proof}

Finally, Dehghan~\cite{Dehghan2016} showed NP-completes of yet another variant of \naeksat{3}, called \restrictedsat.
In this problem we are given an instance $\Phi$ of \textsc{Strongly Planar Positive Not-All-Equal 3-Sat} along with a non-empty subset $\cV'$ of variables, and we ask whether $\Phi$ has a satisfiable truth assignment, such that all variables in $\cV'$ are set $true$. This, along with \autoref{lem:gndi-naesat}, implies the hardness of the following extension variant of our problem.

\begin{theorem}\label{thm:gndi-planar-extend}
Let $G$ be a connected planar bipartite graph with bipartition classes $A,B$, and let $E'$ be a subset of edges of $G$.
It is NP-complete to determine whether $G$ has an edge labeling $c$ with labels $\{1,2\}$, distinguishing neighbors by sets, in which $c(e) = 1$ for every $e \in E'$.
\end{theorem}
\begin{proof}
Let $(\Phi,\cV')$ be an instance of \restrictedsat and $G$ be its incidence graph.
For each variable $v \in \cV'$, we label all edges incident to $v$ with label 1. Note that this means that the variable $v$ must be set $true$. 
Moreover, notice that since $\cV'$ is non-empty, at least one of variable-vertices is forced to be in $X_1$, so there is no ambiguity between the bipartition classes of $G$.
\end{proof}

Finally let us mention another related problem, introduced by Tahraoui, Duch\^ene, and Kheddouci~\cite{TAHRAOUI20123011}.
We say that an edge labeling $c$ of a connected graph $G$ with at least two vertices is an \emph{edge-labeling by gap} if the vertex coloring defined as
\[
f(v) = \begin{cases}
c(e) & \text{ if } \deg v = 1 \text{ and } v \in e,\\ 
\max_{e \ni v} c(e) - \min_{e \ni v} c(e) & \text{otherwise},
\end{cases}
\]
is proper. By $\gap(G)$ we denote the minimum $k$, for which there exists an edge labeling by gap using labels $\{1,2,\ldots,k\}$.
We observe that if $G$ is a bipartite graph with minimum degree at least 2, then $c \colon E \to \{1,2\}$ is an edge-labeling by gap if and only if it distinguishes neighbors by sets.
Indeed, suppose that $c$ is an edge-labeling by gap and consider an edge $uv$ of $G$. Without loss of generality assume that $f(u) < f(v)$, they cannot be equal by the definition of $c$. Since $\deg u, \deg v \geq 2$, and $c$ uses only labels 1 and 2, we must have $f(u)=0$ and $f(v)=1$. This means that $c(E(v))=\{1,2\}$ and $c(E(u)) \in \{ \{1\}, \{2\} \}$, which proves that $c$ is a labeling distinguishing neighbors by sets.
Analogously, suppose $c$ distinguishes neighbors by sets and consider an edge $uv$, such that $u \in X$ and $v \in Y$. By the definition of $c$, we observe that $f(u)=0$ and $f(v)=1$, so $c$ is an edge labeling by gap.

However, the presence of vertices of degree 1 in $G$ changes the situation significantly: Dehghan, Sadeghi, and Ahadi~\cite{DEHGHAN201325} proved that determining if $\gap(G)=2$ is NP-complete for a planar bipartite graph $G$, while it is polynomially solvable if $G$ has minimum degree 2. This is in a contrast with \autoref{thm:gndi-planar-poly}.

\subsection{Graphs $G$ with $\gndi(G)=2$} \label{sec:gndi2}

The \posnaesat is equivalent to a well-known hypergraph problem, whose origins date back to the work of Bernstein~\cite{Bernstein}, and was popularized by  Erd\H{o}s~\cite{Erdos1,Erdos2}. We say that a hypergraph $H$ \emph{has property \B} (or simply is 2-colorable) if there is a 2-coloring of its vertices, in which no edge is monochromatic. It is straightforward to observe that $\Phi$ is a satisfiable \posnaesat formula if and only if the hypergraph whose vertices are variables of $\Phi$, and edges are formed by the clauses, has property \B.
Radhakrishnan and Srinivasan~\cite{RSA:RSA2} proved that every $k$-uniform hypergraph with at most $0.1 \sqrt{k/\ln k} \; 2^k$ edges has property \B. Moreover, if $k$ is large enough, then the bound can be improved to $0.7 \sqrt{k/\ln k} \; 2^k$. By \autoref{lem:gndi-naesat}, we obtain the following corollary.

\begin{corollary}
Let $G$ be a bipartite graph with bipartition classes $A$ and $B$, such that $\deg b = k$ for every $b \in B$.
If $|B| \leq 0.1 \sqrt{k/\ln k} \; 2^k$, then $\gndi(G)=2$.\\
Moreover, if $k$ is sufficiently large and $|B| \leq 0.7 \sqrt{k/\ln k} \; 2^k$, then $\gndi(G)=2$.
\end{corollary}

As the very first application of the Lov\'asz Local Lemma, Erd\H{o}s and Lov\'{a}sz~\cite{erdos1975problems}  proved that each hypergraph, whose every edge has at least $k$ elements and intersects at most $\frac{2^{k-1}}{e}-1$ other edges, has property \B (see also Alon, Spencer~\cite{DBLP:books/wi/AlonS92}). For large $k$ and uniform hypergraphs, this bound was improved by Radhakrishnan and Srinivasan~\cite{RSA:RSA2} who showed that if $H$ is a $k$-uniform hypergraph, in which no edge intersects more than $0.17 \sqrt{k/ \ln k} \; 2^k$ edges, then $H$ has property \B. This gives the following.

\begin{corollary}
Let $G$ be a bipartite graph with bipartition classes $A$ and $B$, such that $\deg b \geq k$ for every $b \in B$.
If for every vertex $b \in B$, the number of vertices at distance exactly 2 from $b$ is at most $\frac{2^{k-1}}{e}-1$, then $\gndi(G)=2$.\\
Moreover, if $\deg b=k$ for every $b \in B$ and $k$ is sufficiently large, then the bound on the number of vertices at distance 2 can be improved to $0.17 \sqrt{k/ \ln k} \; 2^k$.
\end{corollary}

The bound by  Erd\H{o}s and Lov\'{a}sz~\cite{erdos1975problems} implies that if $k \geq 9$, then every  $k$-uniform $k$-regular hypergraph $H$ has property \B. The bound on $k$ can be improved to $k \geq 4$ (the result is attibuted to Thomassen, the proof can be found in Henning, Yeo~\cite{HENNING20131192}). Thus we immediately obtain the following corollary.

\begin{corollary}
For every $k \geq 4$ and every $k$-regular bipartite graph $G$ it holds that $\gndi(G)=2$.
\end{corollary}
 
It is well known that $k \geq 4$ is the best possible bound of this type, as the Fano plane $F$, i.e., the unique $(7,7,3,3)$-configuration, is a 3-uniform 3-regular hypergraph that does not have property \B (see \autoref{fig:fano}).
This is equivalent to saying that the 3-regular incidence graph $G_F$ of $F$ does not have an edge labeling with two labels, that distinguish neighbors by sets (note that by the symmetry of $F$, the bipartition classes of $G_F$ are symmetric).

\begin{figure}[h]
\centering
\begin{subfigure}[b]{0.4\textwidth}
\centering
\begin{tikzpicture}[scale=3,mydot/.style={  draw,  circle,  fill=black,  inner sep=1.5pt}]
\node at (0,0.8) {a)};
\draw   (0,0) coordinate (A) --   (1,0) coordinate (B) --   ($ (A)!.5!(B) ! {sin(60)*2} ! 90:(B) $) coordinate (C) -- cycle;
\coordinate (O) at   (barycentric cs:A=1,B=1,C=1);
\draw (O) circle [radius=1.717/6];
\draw (C) -- ($ (A)!.5!(B) $) coordinate (LC); 
\draw (A) -- ($ (B)!.5!(C) $) coordinate (LA); 
\draw (B) -- ($ (C)!.5!(A) $) coordinate (LB); 
\foreach \Nodo in {A,B,C,O,LC,LA,LB}
  \node[mydot] at (\Nodo) {};      
\end{tikzpicture}
\end{subfigure}
\begin{subfigure}[b]{0.4\textwidth}
\centering
\begin{tikzpicture}[scale=0.5,every node/.style={draw,circle,minimum size = 5pt},red/.style={fill=red!20}, blue/.style={fill=blue!20},
, rb/.style  = {path picture={
    \fill[blue!20](path picture bounding box.north west) -- (path picture bounding box.south east) --  (path picture bounding box.north east)--cycle;
    \fill[red!20](path picture bounding box.north west) -- (path picture bounding box.south west) --  (path picture bounding box.south east)--cycle;
  }}]
  
\node[draw=none] at (-1.5,5) {b)};  
\node (v1) at (0,0) {};  
\node (v2) at (0,1) {};
\node (v3) at (0,2) {};
\node (v4) at (0,3) {};
\node (v5) at (0,4) {};
\node (v6) at (0,5) {};
\node (v7) at (0,6) {};

\node (e1) at (5,0) {};  
\node (e2) at (5,1) {};
\node (e3) at (5,2) {};
\node (e4) at (5,3) {};
\node (e5) at (5,4) {};
\node (e6) at (5,5) {};
\node (e7) at (5,6) {};

\draw (v1) -- (e1);
\draw (v2) -- (e1);
\draw (v3) -- (e1);

\draw (v3) -- (e2);
\draw (v4) -- (e2);
\draw (v5) -- (e2);

\draw (v1) -- (e3);
\draw (v6) -- (e3);
\draw (v5) -- (e3);

\draw (v1) -- (e4);
\draw (v7) -- (e4);
\draw (v4) -- (e4);

\draw (v2) -- (e5);
\draw (v7) -- (e5);
\draw (v5) -- (e5);

\draw (v3) -- (e6);
\draw (v7) -- (e6);
\draw (v6) -- (e6);
 
\draw (v2) -- (e7);
\draw (v4) -- (e7);
\draw (v6) -- (e7); 
\end{tikzpicture}
\end{subfigure}
\caption{a) Fano plane (points denote vertices, straight lines and the circle denote 3-element edges); b) the 3-regular bipartite incodence graph $G_F$ of $F$ with $\gndi(G_F)=3$.} \label{fig:fano}
\end{figure}
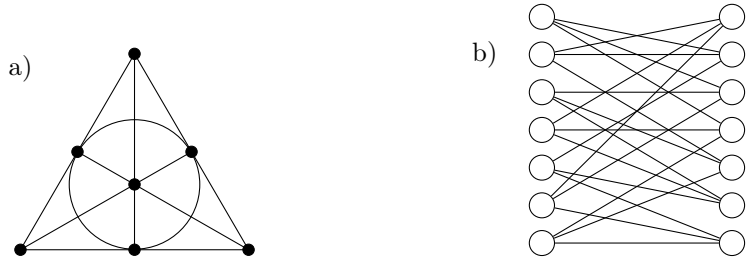

It turns out that the Fano plane is an exceptional case. Person and Schacht~\cite{Person:2009:AHW:1496770.1496795} showed that almost all 3-uniform hypergraphs that do not contain $F$ as a (not necessarily induced) subhypergraph have property \B.
More formally, they showed that if $H$ is a labeled 3-regular hypergraph with $n$ vertices, which does not contain $F$, chosen uniformly at random, then $H$ has property \B with probability at least $1-2^{-\Omega(n^2)}$.
We observe that a hypergraph $H$ contains $F$ as a subhypergraph if and only if its incidence graph contains $G_F$ as an induced subgraph. Thus we obtain the following.

\begin{corollary}
Let ${\cal G}$ be the family of bipartite graphs $G$ with bipartition classes $A,B$, in which $\deg b = 3$ for all $b \in B$, and excluding $G_F$ as an induced subgraph. Then $\gndi(G)=2$ for almost all $G \in {\cal G}$.
\end{corollary}

\section{Conclusions} \label{sec:conclusion}
Let us conclude the paper with pointing out several open questions.
In \autoref{thm:main} we require that $H$ is a uniform hypergraph, which restricts the corollaries in \autoref{sec:corollaries} to regular graphs. Our approach cannot be easily adapted to nonuniform hypergraphs. However, there are some other algorithmic versions of the Lov\'{a}sz Local Lemma that work in such a setting (see e.g. Czumaj and Scheidler \cite{RSA:RSA3}). It would be interesting to investigate if we can find some bounds on the minimum size of lists that guarantee the existence of an ieds- and an iedm-coloring of a nonuniform hypergraph. Observe that in case of iedm-colorings, we will never have a conflict on intersecting edges with different cardinalities.

Next, let us turn our attention to graph labeling problems.
In \autoref{cor:edge-multiset} we show that if $k$ is sufficiently large, then for any $k$-regular graph, lists of size at least $0.37 k$ are sufficient to guarantee the existence of a list edge labeling, distinguishing neighbors by multisets.
On the other hand, we know that 4 colors suffice in the non-list version (and even 3, if the minimum degree is sufficiently large) \cite{ADDARIOBERRY2005237}. Is it possible to get a constant upper bound in the list variant of the problem?

In case of list edge labelings, distinguishing neighbors by sets,  we cannot hope for a constant bound, as for any graph $G$ with $\chi(G) \geq 3$ it holds that $\gndi(G) = \lceil \log_2 \chi(G) \rceil +1$~\cite{DBLP:journals/dm/GyoriP09}. However, it would still be interesting to know whether the bound from \autoref{cor:edge-set} can be improved to a sublinear function of $k$.

Furthermore, as shown in \autoref{sec:gndi2}, a close connection between the property \B of a hypergraph $H$ and the value of $\gndi(G)$ of its incidence graph $G$ yields the existence of neighbor-distinguishig edge labeling using only two labels in several classes of bipartite graphs. However, note that uniformity of $H$, while a natural property of a hypergraph, is translated into a less natural condition of regularity of one bipartition class of $G$. It would be interesting to investigate property \B in hypergraphs, whose incidence graphs belong to some more natural classes of bipartite graphs. Moreover, recall from \autoref{lem:gndi-naesat} that if $G$ is connected and bipartite, then $\gndi(G)=2$ if {\em any} of two hypergraphs, for which $G$ is the incidence graph, has property \B.

Finally, from \autoref{thm:gndi-nphard} and the complexity results mentioned in the introduction it follows that the decision problem of determining whether a hypergraph $H$ with $n$ vertices has an ieds- or an iedm-coloring using at most $c$ colors is NP-hard, even for very restricted classes of hypergraphs. Moreover, assuming the ETH, there is no algorithm solving this problem in time $2^{o(n)}$. This lower bound matches the complexity of the trivial brute force algorithm, if the number of colors $c$ is considered constant. However, if the number of colors is non-constant, the complexity of the brute-force algorithm is $2^{O(n \log n)}$. Can we get a single exponential, i.e., $2^{O(n)}$ time, algorithm for the problem, if the number of colors is arbitrary, or prove that such an algorithm does not exist, under some standard complexity assumptions?

\bigskip
\noindent \textbf{Acknowledgements.}
The authors are sincerely grateful to prof. Zbigniew Lonc, Agnieszka Piliszek, and Bartosz Kołodziejek for fruitful discussions on the topic.

\end{document}